\pdfoutput=1 
\documentclass[conference]{IEEEtran}
\IEEEoverridecommandlockouts
\usepackage{cite}
\usepackage{graphicx}
\usepackage{balance}  % for  \balance command ON LAST PAGE  (only there!)
\usepackage{color}
\usepackage{amsfonts,amssymb,amsmath,bm}
\usepackage{makecell}
\usepackage{pifont}
\usepackage{stfloats}
\usepackage{subfigure}
\usepackage{multirow}
\usepackage[table,xcdraw]{xcolor}
\usepackage{graphicx}
\usepackage[ruled,vlined]{algorithm2e}
\usepackage{verbatim}
\usepackage{adjustbox}
\usepackage{diagbox}
\usepackage{makecell}
\usepackage{booktabs}
\usepackage{url}
\usepackage[normal, footnotesize]{caption}

\usepackage{multirow}
\usepackage[normalem]{ulem}
\usepackage{ulem}
\usepackage{amsthm}
\useunder{\uline}{\ul}{}

\newtheorem{example}{Example}
\newtheorem{definition}{Definition}
 
\newtheorem{lemma}{Lemma} 

\makeatletter
  \newcommand\figcaption{\def\@captype{figure}\caption}
  \newcommand\tabcaption{\def\@captype{table}\caption}
\makeatother

\begin{document}

\title{Accelerating Biclique Counting on GPU}

\author{
Linshan Qiu, Zhonggen Li, Xiangyu Ke, Lu Chen, Yunjun Gao\\

\emph{Zhejiang University, Hangzhou, China} \\
%\\
%\normalsize $^{\sharp}$ \emph{School of Software Technology, Zhejiang University, Ningbo, China}\\
% \normalsize $^{\S}$\emph{Huawei, Chengdu, China}\\

\emph{\{lsqiu, zgli,xiangyu.ke,luchen, gaoyj\}@zju.edu.cn 
}}

\maketitle

\begin{abstract}
Counting $(p,q)$-bicliques in bipartite graphs poses a foundational challenge with broad applications, from densest subgraph discovery in algorithmic research to personalized content recommendation in practical scenarios. 
Despite its significance, current leading $(p,q)$-biclique counting algorithms fall short, particularly when faced with larger graph sizes and clique scales. 
Fortunately, the problem's inherent structure, allowing for the independent counting of each biclique starting from every vertex, combined with a substantial set intersections, makes it highly amenable to parallelization.
Recent successes in GPU-accelerated algorithms across various domains motivate our exploration into harnessing the parallelism power of GPUs to efficiently address the $(p,q)$-biclique counting challenge.  
%\textcolor{red}{[XY: one sentence about: the problem could be parallelize, hence suitable to solve via GPUs.]}

We introduce \textsf{GBC} (\underline{\textsf{G}}PU-based \underline{\textsf{B}}iclique \underline{\textsf{C}}ounting), a novel approach designed to enable efficient and scalable $(p,q)$-biclique counting on GPUs. 
%To overcome low thread utilization arising from DFS exploration in counting $(p,q)$-bicliques, 
%First, we present an innovative hybrid DFS-BFS exploration strategy that enhances thread utilization while effectively managing memory constraints. 
%Subsequently, a composite load balancing strategy, integrating static and dynamic workload allocation, is implemented to equitably distribute the workload among threads.
To address major bottleneck arising from redundant comparisons in set intersections (occupying an average of 90\% of the runtime), we introduce a novel data structure that hashes adjacency lists into truncated bitmaps to enable efficient set intersection on GPUs via bit-wise AND operations.
%while preventing redundant comparisons.
Our innovative hybrid DFS-BFS exploration strategy further enhances thread utilization and effectively manages memory constraints. A composite load balancing strategy, integrating pre-runtime and runtime workload allocation, ensures equitable distribution among threads.
Additionally, we employ vertex reordering and graph partitioning strategies for improved compactness and scalability. 
Experimental evaluations on eight real-life and two synthetic datasets demonstrate that \textsf{GBC} outperforms state-of-the-art algorithms by a substantial margin. In particular, \textsf{GBC} achieves an average speedup of 497.8$\times$, with the largest instance achieving a remarkable 1217.7$\times$ speedup when $p=q=8$. 
\end{abstract}


\section{Introduction}
\label{introduction}

%Bipartite graphs are extensively employed to model the relationships between two distinct sets of entities, such as social networks~\cite{wang2021efficient}, recommender systems~\cite{GraphJet}, and e-commerce networks~\cite{li2020hierarchi}, to name a few.
%Figure~\ref{fig:example and time breakdown}(a) exemplifies a case of a bipartite graph $G(U, V, E)$, where $U$ and $V$ represent two disjoint vertex sets, and $E$ is the edge set, where edges are restricted to connect vertices belonging to distinct sets. Let us consider a specific instance in the context of a recommender system, the vertices of the bipartite graph represent users and items (e.g., movies, products), while the edges signify interactions, like user-item purchases or ratings. 
%Given its prevalence in real-world applications, considerable efforts have been devoted to proposing diverse algorithms catered to various tasks on bipartite graphs, including community search~\cite{wang2021efficient}, cohesive subgraph discovery~\cite{liu2019efficient,wang2020efficient}, butterfly counting~\cite{sanei2018butterfly}, graph learning~\cite{yang2022scalable,li2020hierarchical}, etc. 

Bipartite graphs are pivotal for illustrating connections between two distinct sets of entities, finding practical applications across diverse domains such as social networks~\cite{wang2021efficient}, recommender systems~\cite{GraphJet}, and e-commerce networks~\cite{li2020hierarchi}. 
In this model, one set signifies a specific type of entity (e.g., users), while the other set represents a different type (e.g., items). Edges establish connections from entities in the first set only to those in the second, capturing relationships or interactions between them.
There has been a wide range of explorations over bipartite graphs, such as community search~\cite{wang2021efficient}, cohesive subgraph discovery~\cite{liu2019efficient, wang2020efficient}, butterfly counting~\cite{sanei2018butterfly}, and graph learning~\cite{yang2022scalable, li2020hierarchi}. 
However, within this intricate landscape, a critical challenge emerges—the enumeration of $(p,q)$-bicliques, a combinatorial task with far-reaching implications. A $(p,q)$-biclique represents a complete subgraph with $p$ vertices from one set and $q$ vertices from the other. Notably, the well-known butterfly concept corresponds to the $(2,2)$-biclique~\cite{yang2021p}.
The enumeration of $(p,q)$-bicliques holds crucial importance in both algorithmic research and various applications, {including densest subgraph detection~\cite{mitzenmacher2015scalable}, cohesive subgroup analysis~\cite{borgatti1997network} in bipartite graphs, and optimization of GNN information aggregation~\cite{yang2021p}. 
For instance, the pursuit of the $(p,q)$-biclique densest subgraph in a bipartite graph relies on the concept of $(p, q)$-biclique density~\cite{mitzenmacher2015scalable}, which is defined as the ratio between the number of $(p, q)$-bicliques in a subgraph $S$ and the size of $S$.
}
% For instance, in content recommendation systems, $(p,q)$-bicliques identify user clusters and relevant content groups, aiding personalized recommendations based on the strength of group associations.
%For instance, in content recommendation systems, $(p,q)$-bicliques can pinpoint clusters of users with shared interests and groups of content items tailored to those interests. The quantity of $(p,q)$-bicliques involving a user group serves as an indicator of the strength of association within that group, facilitating more personalized content recommendations.

\begin{figure}[tbp]
    \centering
    \includegraphics[width=0.45\textwidth]{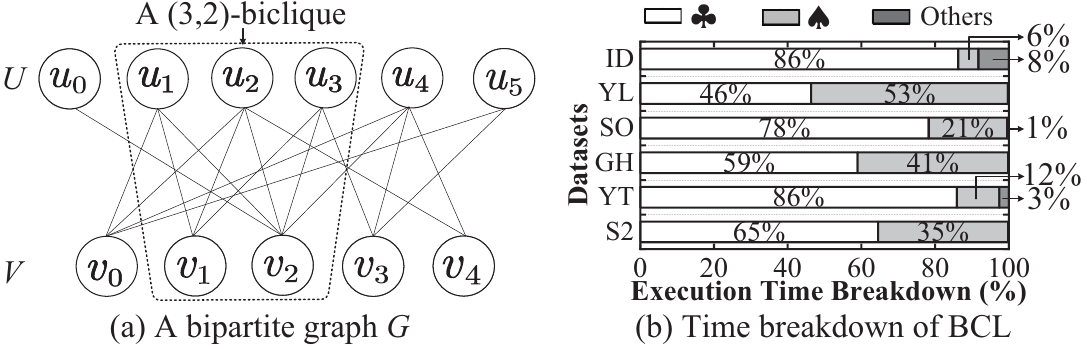}
    \vspace{-2mm}
    \caption{An example of bipartite graph and time breakdown of \textsf{BCL} ($\clubsuit$ and $\spadesuit$ denotes searching for shared 1-hop and 2-hop neighbors, respectively).}
    \label{fig:example and time breakdown}
    \vspace{-8mm}
\end{figure}

%\textcolor{red}{[XY: use a example block to illustrate fig 1]}

%In this paper, we study the problem of ($p,q$)-biclique counting, where a ($p,q$)-biclique is a complete subgraph $B(L, R)$ of a bipartite graph with $|L|=p$, $|R|=q$, where $L\subseteq U$, $R\subseteq V$ and $\forall (u,v)\in L \times R$, $(u,v)\in E$. In Figure~\ref{fig:example and time breakdown}(a), for example, vertex sets $L=\{u_1,u_2\}$, $R=\{v_1,v_2,v_3\}$ and the edges connecting vertices in $L$ and $R$ form a $(2,3)$-biclique. 
%The concept of a $(p, q)$-biclique represents a broader and more inclusive form of biclique-based cohesion compared to other specific counterparts, such as the butterfly, which corresponds to a $(2,2)$-biclique~\cite{yang2021p}.
%Furthermore, the quantity of $(p,q)$-bicliques in a bipartite graph serves as a pivotal foundation for graph analysis and various applications. For example, in content recommendation systems, $(p,q)$-bicliques can be used to identify groups of users who share specific interests and groups of content items that cater to those interests. The number of $(p,q)$-bicliques involving a group of users serves as an indicator of the strength of association among this user group, fostering more personalized content recommendations.

Counting $(p,q)$-bicliques presents a formidable challenge, given its exponential increase concerning $p$ and $q$~\cite{yang2021p}. %This renders it a computationally intensive task~\cite{yang2021p}. 
% {Fundamentally, retrieving $(p,q)$-bicliques is actually a collection of alternative intersections to find common neighbors of the vertices on each side in the partial bicliques. These common neighbors are then added to expand the partial bicliques until the size constraints are met. Therefore, identifying common neighbors (intersection of adjacency lists) contributes to the foundation of the biclique counting algorithm.}
%Especially, retrieving $(p,q)$-bicliques requires iteratively finding the 1-hop and 2-hop\footnote{Those vertices share at least $p$ or $q$ neighbors in the other side.} common neighbors of each vertex in the partial results, i.e., \emph{intersections between the adjacency lists of the vertices under exploration and those in the partial results}.
%{Especially, retrieving $(p,q)$-bicliques requires iteratively finding the common 1-hop neighbors of the vertices that are 2-hop neighbors to each other, i.e., intersections between the adjacency lists of the vertices under exploration and those in the partial results. }
Specifically, retrieving $(p,q)$-bicliques involves iteratively identifying the common 1-hop neighbors for a set of vertices that are mutually 2-hop neighbors\skip\footins\smallskipamount\footnote{A pair of vertices share at least $p$ or $q$ neighbors on the other side.}, requiring intersections between the adjacency lists of the vertices under exploration and those in the partial results.
Yang et al. undertake the pioneering investigation and introduce the leading algorithm in a backtracking manner, namely \textsf{BCL}~\cite{yang2021p}.
%{\{LS: how about not much detail of BCL, while focusing on the general method above.\}}
% \textsf{BCL} computes the qualified ($p,q$)-biclique by means of backtracking enumeration. Specifically, \textsf{BCL} first selects a layer, assuming it's layer $U$, to construct a 2-hop graph $H$\footnote{The vertex set of $H$ is $U$, and edges are established between pairs of vertices in $U$ if they share at least $q$ vertices on the opposing layer $V$.} that optimizes the computation overhead according to the designed cost function. Subsequently, the algorithm recursively searches for (partial) $p$-cliques $L$ in $H$ {by extracting subgraph $H^\prime$ induced by the neighbors of the newly added vertex}, while identifying the shared neighbor sets $S$ among the vertices comprising the ongoing (partial) $p$-clique in $G$. 
%\textsf{BCL} employs backtracking enumeration to compute the qualified $(p,q)$-bicliques. 
% {\textsf{BCL} initially chooses a vertex set (say $U$) that minimizes the overall complexity to construct a graph connecting vertices that share at $q$ 1-hop neighbors (on the opposing layer). Subsequently, it recursively explores sub-cliques (say $p$-cliques) in this constructed graph, while identifying the common neighbors among the vertices of the sub-cliques on the opposing layer.}
However, \textsf{BCL} encounters scalability issues\footnote{On real-world million-scale dataset \textit{FR}, the running time of \textsf{BCL} exceeds 24 hours when $p=q=8$. Please refer to \S~\ref{sec:overallperformance} for more details.} concerning either dataset size or {clique scale, i.e., larger $p$ and $q$}. 
We observe that the inefficiency in identifying shared 1-hop and 2-hop neighbors via intersections is the primary culprit. 
%, becoming the bottleneck of the algorithm.
Figure~\ref{fig:example and time breakdown}(b) visually breaks down the execution time taken by \textsf{BCL} across six selected datasets, with labels omitted for clarity when bar heights fall below 1\%. As depicted, the time devoted to searching for shared 1-hop and 2-hop neighbors peaks at more than 99\%, averaging at 97\%.
Consequently, there is an urgent need to optimize the intersection to enhance the algorithm's efficiency. 

\vspace{-1.5mm}
\begin{example}
    Consider the identification of $(3,2)$-bicliques in Figure~\ref{fig:example and time breakdown}(a), the dotted line circle presents an instance with vertices $\{u_1,u_2,u_3\}$ from layer $U$ and $\{v_1,v_2\}$ from layer $V$. Notably, $\{u_2, u_3, u_4\}$ are also mutual 2-hop neighbors. Specifically, $u_2$ and $u_3$ share 1-hop neighbors $\{v_1, v_2\}$, $u_2$ and $u_4$ have 1-hop neighbors $\{v_0, v_2, v_4\}$, and $u_3$ and $u_4$ possess 1-hop neighbors $\{v_2, v_3\}$. However, despite these mutual 2-hop connections, they only share a common 1-hop neighbor, $v_2$, preventing them from forming a $(3,2)$-biclique. The set of vertices with a mutual 2-hop neighborhood relationship serves as candidates, and the fundamental test involves checking their common neighbors, leading to exponential growth in computational costs concerning the target clique size $(p,q)$.
\end{example}
\vspace{-1.5mm}

As real-world graphs undergo exponential growth, their computational demands have intensified. A recent and notable trend involves the utilization of GPUs to streamline a variety of graph algorithms, including shortest path~\cite{lu2020accelerating}, PageRank~\cite{shi2019realtime}, breadth-first traversal~\cite{liu2015enterprise}, graph pattern mining~\cite{chen2020pangolin,wang2016gunrock}, leading to remarkable performance advancements. 
Compared to CPUs, GPUs stand out with numerous computation cores and high-bandwidth memory, making them ideal for computationally intensive tasks. Encouraged by the recent success in GPU-accelerated graph algorithms, 
%\textcolor{red}{[XY: this is not clear here. the current paragraph is about why GPUs are suitable for solving the problem. I dont see anything about combining CPUs and GPUs]}{parallel acceleration on GPU platform} 
GPU-based parallelism emerges as a promising solution to expedite $(p,q)$-clique counting. 
The independence of counting $(p,q)$-bicliques from each vertex presents significant potential for parallelizing $(p,q)$-biclique counting. Furthermore, implementing the algorithm's predominant procedures through intersections, feasibility proven on GPUs~\cite{xu2022efficient, Hu0L21}, enhances its efficiency.

\textbf{Challenges.} GPUs operate in an SIMT (Single Instruction, Multiple Threads) manner, which is distinct from the typical CPU paradigm. A simple transposition of the existing algorithm to the GPU, as evidenced in \S\ref{sec:evaluation}, does not yield satisfactory performance improvements. Specifically, three primary challenges are hindering the development of a high-performance $(p,q)$-biclique counting algorithm on the GPU.

\textit{\textbf{Challenge \uppercase\expandafter{\romannumeral1}}: How to implement efficient set intersection on GPU?} 
The inefficiency of set intersection in existing GPU-based algorithms arises from redundant comparisons and substantial memory access{~\cite{pandey2021trust}}. Various efforts within the field of triangle counting~\cite{hu2018tricore, Hu0L21, BissonF17} have sought to optimize set intersection on GPUs, with the binary search-based approach{~\cite{Hu0L21}} emerging as a leading efficient technique.
Performing binary searches on adjacency lists involves element-wise comparisons with elements residing in global memory, increasing computational overhead and memory access latency—particularly when dealing with lengthy adjacency lists (larger graphs) and probing deeper search trees (larger clique size).
To address these challenges, we introduce a novel data structure, Hierarchical Truncated Bitmap, which hashes vertices in the adjacency list into 32-bit truncated bitmaps (integers) and aggregates the offsets of these bitmaps as a range index. This enhancement facilitates set intersection through bit-wise AND operations, mitigating redundant comparisons (\S~\ref{sec:bitmap}).
Furthermore, we propose a vertex reordering method named \textsf{Border}, aiming to maximize the storage of multiple vertices within a single bitmap, thus effectively compressing the data and reducing data retrieval overhead (\S~\ref{sec:vertexreorder}).

{\textit{\textbf{Challenge \uppercase\expandafter{\romannumeral2}}: How to adapt the biclique counting algorithm to align with GPU architecture?}} 
The GPU, with its abundance of available threads numbering in the thousands, demands optimal utilization to fully harness its computational prowess.
However, several issues contribute to low GPU thread utilization. 
Firstly, as exploration depth increases, the elements pending examination in the partial result sets progressively diminish. Using a fixed-size thread group leads to a significant portion of threads being idle. {Additionally, traversing vertices through backtracking exploration one at a time results in thread and bandwidth wastage.}  
To mitigate these concerns, we diverse a hybrid DFS-BFS search strategy to adaptively unify the tasks of multiple vertices for enhanced parallelism and thread utilization (\S~\ref{sec:bicliquecounting}). 
Secondly, the skewed distribution of vertex degrees and the unpredictable workload during depth-first traversal lead to a marked imbalance during runtime.
%distribution of tasks and the inadequacy of workload distribution before runtime. 
Consequently, we employ a combination of pre-runtime and runtime task distribution strategies to achieve equilibrium in the distribution of workload among threads (\S~\ref{sec:workloadbalance}).

{\textit{\textbf{Challenge \uppercase\expandafter{\romannumeral3}}: How to improve the scalability for large graphs?}} 
The limited memory capacity of GPUs presents a challenge in accommodating large-scale graphs within device memory, typically ranging from a few to tens of gigabytes. Real-world graphs, however, expand exponentially, posing difficulties for complete residence in GPU memory~\cite{guo2020gpu}. 
Various efforts have addressed limited device memory through graph partitioning~\cite{guo2020gpu,pandey2021trust}. Yet, due to partition interdependence, frequent loading and eviction may occur, resulting in significant transmission overhead. Moreover, a substantial portion of imported data remains unutilized. This inefficiency, coupled with computation cores waiting for required data transfers, impedes throughput.
Leveraging the fact that clique computations involve interactions with, at most, 2-hop neighbors, we partition the graph into disjoint and autonomous subgraphs. This approach allows computations for any given vertex to occur within the confines of the current subgraph, eliminating the need for on-demand data loading (\S~\ref{sec:outofmemory}).

\textbf{Contributions.} In this work, we introduce \textsf{GBC} (\underline{\textsf{G}}PU-based \underline{\textsf{B}}iclique \underline{\textsf{C}}ounting), the pioneering work facilitating $(p,q)$-biclique counting by harnessing the massive parallelism of GPUs. We outline the key advancements as follows:

\vspace{-0.5mm}
\begin{itemize}
    \item We devise a novel \emph{data structure} and \emph{vertex reordering} technique to implement highly efficient intersection computation on GPUs.
    \item We advocate the adoption of a \emph{hybrid search} strategy to optimize thread utilization. Additionally, our \emph{joint load-balancing} strategy optimizes pre-runtime and runtime workload allocation for improved efficiency.
    \item We propose a \emph{communication-free partitioning} method, enabling the computation of $(p,q)$-bicliques on large-scale graphs using GPUs. 
    \item We conduct \emph{comprehensive experiments} across ten diverse datasets to demonstrate the superior performance of the proposed approach. The results showcase an average speedup of over 400$\times$ compared to existing baselines. 
\end{itemize}
\vspace{-0.5mm}

\textbf{Roadmap.} We formally define the problems in \S~\ref{sec:preliminaries}, followed by a brief introduction to the leading CPU-based solution and our basic implementation on GPUs in \S~\ref{sec:sota}. \S~\ref{sec:bicliquecounting} illustrates the search paradigm, and \S~\ref{sec:optimizations} presents various optimization techniques, including data structure, vertex reordering, and load balancing. Scalability consideration is addressed in \S~\ref{sec:outofmemory}. Experiments are presented in \S~\ref{sec:evaluation}. We review related works in \S~\ref{sec:relatedwork} and conclude the paper in \S~\ref{sec:conclusion}.

\section{Preliminaries}
\label{sec:preliminaries}
% \vspace{-1mm}

In this section, we provide the formal problem definition and give a brief introduction to GPUs. {The frequently used notations are summarized in Table~\ref{tab:symbols}}.

% \vspace{-1mm}
\subsection{Problem Definition}
Given an unweighted and undirected bipartite graph $G = (U, V, E)$, $U(G)$ and $V(G)$ denote two disjoint sets of vertices on the upper and lower layers, respectively (i.e., $U(G)\cap V(G)=\emptyset$). $E(G)\subseteq U(G)\times V(G)$ represents the edge set of $G$, where a edge $(u,v)$ can only exist between $u\in U(G)$ and $v\in V(G)$. {We use $N(u, G)=\{v|(u,v)\in E(G)\}$ and $N_{2}(u,G)=\{u'|(u,v)\in E(G)\wedge (u',v)\in E(G)\}$ to denote the 1-hop neighbors (vertices directly connected to $u$) and 2-hop neighbors (vertices indirectly connected to $u$ via an intermediate vertex) of vertex $u$ in $G$, respectively.} 
{Additionally, we extend this notion to encompass the 2-hop neighbors sharing at least $k$ ($k=p$ or $q$ in our problem) common 1-hop neighbors with $u$, i.e., $N^k_2(u, G)= \{u^\prime|u^\prime \in U\cup V \text{and } |N(u, G) \cap N(u^\prime, G)| \geq k \}$. 
By default, when referring to 2-hop neighbors, we specifically indicate $N^k_2(u, G)$, particularly during intersection operations.}
% {Additionally, we extend this notion to encompass the 2-hop neighbors sharing at least $k$ ($k=p$ or $q$) common 1-hop neighbors with $u$, i.e., $N^k_2(u, G)= \{u^\prime|u^\prime \in U\cup V \text{and } |N(u, G) \cap N(u^\prime, G)| \geq k \}$.}
% {By default, the 2-hop neighbors refer to $N^k_2(u, G)$.} 
The degree of $u$ is defined as $d(u,G)=|N(u,G)|$. For simplicity, the notation omits $G$ when the context is self-evident.

% \vspace{-1mm}
\begin{definition}[\textbf{(\textit{p,q})-Biclique}~\cite{yang2021p}]
    Given a bipartite graph $G = (U, V, E)$, a \textbf{biclique} $B(L, R)$ is a complete bipartite subgraph of $G$, where $L\subseteq U(G)$, $R\subseteq V(G)$, and $\forall (u,v)\in L\times R$, $(u,v)\in E(G)$. A biclique $B(L, R)$ is a \textbf{(p,q)-biclique} if it satisfies $|L|=p$ and $|R|=q$, where $p$, $q$ are two integers.
\end{definition}
% \vspace{-1mm}

% \vspace{-2mm}
\begin{example}
    {In Figure~\ref{fig:example and time breakdown}(a), two (3,2)-bicliques exist within $G$: $B_1(L_1,R_1)$, where $L_1=\{u_1,u_2,u_3\}$ and $R_1=\{v_1,v_2\}$, and $B_2(L_2,R_2)$, where $L_2=\{u_1,u_2,u_4\}$ and $R_2=\{v_0,v_2\}$.
    As the $(p,q)$-biclique constitutes a complete bipartite subgraph of $G$, the primary computational workload involves identifying the common neighbors of the partial result, essentially the intersection operations.
    }
\end{example}
% \vspace{-2mm}

% \begin{definition}[\textbf{$(p,q)$-Biclique}~\cite{yang2021p}]
%     Given a bipartite graph $G$ and two integer $p$, $q$, a biclique $B(L, R)$ of $G$ is a \textbf{$(p, q)$-biclique} if it satisfies $|L|=p$ and $|R|=q$. 
% \end{definition}

% \begin{definition}[\textbf{$\tau$-Strength 2-Hop Neighbor}~\cite{yang2021p}]
%     Given a bipartite graph $G = (U, V, E)$ and an integer $\tau$, the $\tau$-strength 2-hop neighbors of $w$ in $G$, denoted as $N^\tau_2(w, G)$, covers all vertices in $G$ with at least $\tau$ common neighbors with $w$, i.e., $N^\tau_2(w, G)= \{w^\prime|w^\prime \in U\cup V \text{and } |N(w, G) \cap N(w^\prime, G)| \geq \tau \}$\footnote{For simplicity, and without ambiguity, for $u\in  U$, we refer to $q$-strength neighbors as 2-hop neighbors $N_2(u, G)$, i.e., $N_2(u, G) = N^q_2(u, G)$.}.
% \end{definition}

% \begin{definition}[\textbf{2-Hop Graph}~\cite{yang2021p}]
%     Given a bipartite graph $G = (U,V,E)$, and a pair of parameters $p$ and $q$, the 2-hop graph $H=(U, E)$ of $G$ is a graph induced by $G$ with satisfying: (i) $U(H):=U(G)$, and (ii) $\forall u, v \in U(H), (u,v) \in E(H)$ if $u$ and $v$ are 2-hop neighbors in $G$.
% \end{definition}

\textbf{Problem Statement.} Given a bipartite graph $G$ and two integers $p, q$, the problem of \textbf{\textit{biclique counting}} aims to determine the cardinality of $(p, q)$-bicliques of $G$.

% \vspace{-2mm}
\subsection{GPU Architecture}
\label{gpuarch}
% \vspace{-1.5mm}

GPU is a high-performance hardware with tens of streaming multiprocessors (SMs), each functioning as an autonomous processing unit equipped with hundreds of cores. 
% For programming, CUDA (Compute Unified Device Architecture) abstracts GPU architecture, acting as a bridge between an application and its potential GPU implementation.
{The CUDA (Compute Unified Device Architecture) programming model organizes 32 threads into a warp, where threads execute in a Single Instruction Multiple Thread (SIMT) manner.} Hence, program branches would lead to thread divergence and degenerate performance. 
In CUDA, a block consists of multiple warps and is assigned to an SM. All blocks collectively form a grid with all GPU threads.

GPU typically has multiple levels of memory hierarchy. 
Global memory, shared among all threads, is the largest GPU memory reservoir. Reading (writing) data from (to) it incurs significant latency.
In CUDA, all threads within a warp access the consecutive addresses in the global memory, which is known as coalesced memory access. Therefore, accessing non-consecutive memory for threads in the same warp necessitates multiple memory transactions, leading to low bandwidth. 
Each SM has a fast, small shared memory (typically 16KB to 64KB per SM) exclusively accessible by threads within a block.

\begin{table}[t]
	\vspace*{-0.1in}
	\centering
	\setlength{\tabcolsep}{2pt}
	\caption{{Symbols and descriptions.}}
	\vspace{-0.1in}
	\begin{tabular}{p{2.4cm}p{6cm}}
		\toprule[0.8pt]
		{\textbf{Symbols}} & {\textbf{Description}} \\ \midrule
		{$G$} & {the bipartite graph} \\ \midrule
		{$U(G) (V(G))$, $E(G)$} & {the upper (lower) layer vertex set and the edge set of $G$}\\ \midrule
		{$N(u,G)$, $N_2(u,G)$, $N^k_2(u,G)$} & {1-hop neighbors, 2-hop neighbors of $u$, and 2-hop neighbors sharing at least $k$ 1-hop neighbors with $u$}\\ \midrule
		{$B(L,R)$} & {the biclique $B$ with vertex sets $L$ and $R$} \\ \midrule
		{$p, q$} & {the sizes of $L$ and $R$, respectively} \\ \midrule
        {$C_L, C_R$} & {the candidate sets for $L$ and $R$, respectively}\\ 
        
		\bottomrule[0.8pt]
	\end{tabular}\label{tab:symbols}
	\vspace*{-0.1in}
\end{table}

% \vspace{-2mm}
\section{Base Model \& GPU Baseline }
% \vspace{-2mm}

\begin{figure}[tbp]
    \centering
    \includegraphics[width=0.489\textwidth]{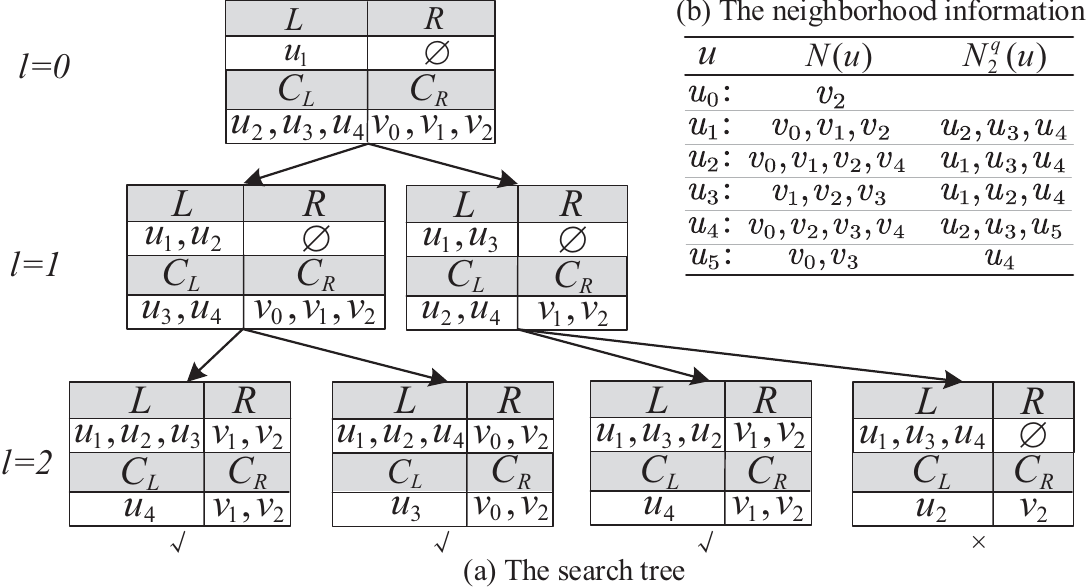}
    \vspace{-6mm}
    \caption{{A walk-through example of basic model (\textsf{Basic}).}}
    \label{basic_example}
    \vspace{-6mm}
\end{figure}

\label{sec:sota}
In this section, we first discuss the basic model and leading solutions on CPU, followed by our baseline design on GPU. Then we give an overview of our proposed methods.

% \vspace{-2mm}
\subsection{Basic Model for Biclique Counting}
% \vspace{-1.5mm}

The basic model (referred to \textsf{Basic}) for biclique counting first selects one side of the vertices as the start vertices (say $U$)\footnote{Without loss of generality, assuming \textsf{Basic} consistently opts for layer $U$ as the starting layer.}. Subsequently, it assembles the 2-hop neighbors of these vertices that share at least $q$ 1-hop neighbors, constituting the candidate set for this layer (denoted as $C_L$). 
Through iterative procedures, candidates in $C_L$ are incrementally added to $L$ to expand the partial result on layer $U$, and $C_L$ is updated by intersecting with the 2-hop neighbors of vertices in $L$. 
Simultaneously, the candidate set $C_R$ undergoes updates by intersecting with 1-hop neighbors of vertices in $L$. 
After iterating the above steps $p$ times, we form a $(p,q)$-biclique $B(L,R)$ by selecting $q$ vertices (if available) from $C_R$ to $R$ and combining them with $L$. 
Otherwise, \textsf{Basic} backtracks and explores adding other candidates to form a $(p,q)$-biclique. 
The following example illustrates the workflow of \textsf{Basic}.

% \vspace{-1mm}
\begin{example}
    Given $G$ in Figure~\ref{fig:example and time breakdown}(a), and two parameters \textit{p=3,q=2}, the workflow of the \textsf{Basic} to retrieve $(3,2)$-biclique starting from $u_1$ is presented in Figure~\ref{basic_example}(a). The neighborhood information of $G$ is summarized in Figure~\ref{basic_example}(a). 
    At level 0, \textsf{Basic} initializes $L$ by adding $u_1$ and establishes $C_L,C_R$ as $N_2^q(u_1)$ and $N(u_1)$, correspondingly.
    With the inclusion of $u_2$ in $L$, \textsf{Basic} traverses the left branch from level 0 to level 1. Subsequently, \textsf{Basic} updates $C_L$ and $C_R$. Specifically, we have $C_L=C_L\cap N_2^q(u_2)=\{u_3,u_4\}$ and {$C_R=C_R\cap N(u_2)=\{v_0, v_1, v_2\}$}.
    \textsf{Basic} advances to the leftmost leaf at level 2 by appending $u_3$ to $L$ and updating $C_L$ and $C_R$. Consequently, a $(3,2)$-biclique, i.e., $(\{u_1,u_2,u_3\},\{v_1,v_2\})$ is found by selecting 2 vertices from $C_R$ to $R$ in conjunction with $L$.
    Similarly, another two $(3,2)$-bicliques can be identified, namely $(\{u_1,u_2,u_4\}$ and $\{v_0,v_2\})$ and $(\{u_1,u_3,u_2\},\{v_1,v_2\})$ (a duplicate).
\end{example}
% \vspace{-2mm}

\textbf{Set intersection.}
% As depicted in Figure~\ref{fig:example and time breakdown}(b), the intersection operation constitutes the primary workload of Figure~\ref{basic_example}. 
In CPU implementations, the common practice for intersection operation in Figure~\ref{basic_example} is linear search~\cite{yang2021p}. 
Conversely, on the GPU, binary search is typically used~\cite{Hu0L21}, with one set as the search key and the other as the search list.
However, existing binary search-based methods suffer from notable computational overhead and memory reads (see \S V-A for details).

{\textbf{State of the art solution on CPU.}
Building upon \textsf{Basic}, \textsf{BCL} employs recursion to continuously select vertices from candidate sets to expand partial results. 
\textsf{BCL} further minimizes the time complexity by utilizing the degree information to select the starting layer.
To reduce intermediate results, \textsf{BCL} employs preallocated arrays and vertex labeling techniques, replacing frequent array creation with array element switching.
Yang et al.~\cite{yang2021p} parallelize \textsf{BCL}, resulting in \textsf{BCLP}, where vertices from the selected layer are distributed to different CPU threads, and each thread executes the \textsf{BCL} algorithm.
Nevertheless, it encounters performance issues outlined in~\S~\ref{introduction} when confronted with dataset size or biclique scale.}

% \textbf{State of the art solution on CPU.} Building upon \textsf{Basic}, Yang et al.~\cite{yang2021p} introduce \textsf{BCL}, which utilizes the degree information of layers to select the layer that minimizes the time complexity as the starting layer. Nevertheless, it suffers from the performance issues mentioned in~\S~\ref{introduction} confronted with dataset size or biclique scale.

Running \textsf{Basic} in parallel on a GPU is straightforward by distributing tasks of different vertices to distinct threads (groups). However, simply porting it to the GPU does not fully harness the high-performance capabilities of the processor. 
In~\S~\ref{gpubaseline}, we present the trivial GPU baseline derived from \textsf{Basic} and elaborate on why its efficiency is not fully realized, despite utilizing such a high-performance processor. 

% \vspace{-1.5mm}
\subsection{Baseline for GPU Implementation}
\label{gpubaseline}
% \vspace{-1.5mm}

The exploration initiated from a vertex $u$ establishes a search tree with $u$ as its root. The search trees of distinct vertices are mutually independent. 
To exploit parallelism, we allocate the vertices in the selected layer to various thread blocks. These blocks then autonomously undertake exploration to retrieve bicliques for their allocated vertices in a backtracking manner. 
% It's worth noting that although breadth-first search (BFS)  might be more suitable for parallelism, our choice of DFS is necessitated by limited GPU memory, as we cannot afford to allocate numerous arrays for storing the explosive partial results. 
Notably, vertices with 2-hop neighbors less than $p - 1$ are not allocated. This exclusion is due to the impossibility of finding a $(p, q)$-biclique on the search trees rooted at these vertices. 
To avoid duplicate results (e.g., the 1st and 3rd leaf nodes of the search tree in Figure~\ref{basic_example}), we assign a unique vertex priority (Definition~\ref{def:priority}) for the vertices on the selected layer and traverse vertices from high priority to low priority. The vertex priority defined in Definition~\ref{def:priority} prevents concentration of computation on high-degree vertices caused by the power-law distribution of vertex degrees, thereby achieving a more balanced workload. Moreover, neighbors with lower priority are not stored to reduce memory overhead. 

% \vspace{-1mm}
\begin{definition}[\textbf{Vertex Priority}]
\label{def:priority}
For any vertex $u$ in the selected layer (suppose $U$) of the bipartite graph $G$, the priority $\mathcal{P}(u)$ is an integer in $[1,|U|]$. For any two vertices $u$ and $w$ in $U$, $\mathcal{P}(u)>\mathcal{P}(w)$ if:

\begin{enumerate}
\vspace{-0.5mm}
    \item $|N_2^q(u)| < |N_2^q(w)|$, or
    \item $|N_2^q(u)| = |N_2^q(w)|$ and $id(u)<id(w)$
\vspace{-0.5mm}
\end{enumerate}

where $id(u)$ is the unique vertex ID of $u$
\end{definition}
% \vspace{-1mm}

% We are accustomed to using recursion to travel a tree by materializing \textsf{Basic} on GPUs, which is infeasible. 
{Directly using recursion to materialize Basic on GPUs is infeasible.}
On one hand, managing memory distribution becomes challenging, and memory consumption is high when recursion is employed on GPUs~\cite{almasri2022parallel}. On the other hand, recursion makes it difficult to share intermediate results among threads.
Therefore, we opt for iteration and use arrays $C_R$ and $C_L$ to respectively store intersection results in each level for backtracking.
Nodes in the search trees entail two computationally intensive operations:
% \textcolor{red}{[LS: clarify why using array $C_L,C_R$ for each layer instead of a single $C_L,C_R$ for all layers, the conversion of the base model on CPU to GPU?]}
% \footnote{We use node to distinguish the vertex of graph.}  

% \vspace{-1mm}
\begin{enumerate}
    \item Intersect $C_L[l-1]$ and $N_2^q(u)$ for $C_L[l]$;
    \item Intersect $C_R[l-1]$ and $N(u)$ for $C_R[l]$.
\end{enumerate}
% \vspace{-1mm}

% Essentially, constructing the new subgraph $H^\prime$ involves identifying the common vertices between current $H^\prime$ and $N(u,H)$. In our GPU implementation, {we leverage intersection to update $H^\prime$, and use arrays to store $H^\prime$ for each layer, i.e., $H^\prime[l] = H^\prime[l-1]\cap N(u,H)$}. 
\textit{These two operations unify the processes of computing the candidate sets for both $L$ and $R$ into intersection calculations, holding significant potential for parallel execution, effectively harnessing the substantial parallelism capabilities of GPU.}
% As a consequence, we utilize multiple warps to concurrently perform set intersection operations. 

We employ parallel binary search for conducting intersections on GPUs, a technique well-regarded for its efficiency~\cite{Hu0L21}. With parallel binary search, threads organized within a warp retrieve an element from the sorted set $C_L[l-1]$ and perform searches against the elements within the sorted set $N_2^q(u)$. The rationale behind fetching elements from $C_L[l-1]$ rather than $N_2^q(u)$ is rooted in the fact that the size of $C_L[l-1]$ often remains smaller than that of $N_2^q(u)$. Likewise, we adopt a similar approach for computing the intersection results between  $C_R[l-1]$and $N(u)$. 

Nonetheless, transferring the algorithm directly to the GPU leads to inefficiencies. In empirical experiments, there are instances where the GPU implementation, despite employing thousands of threads, exhibits slower performance compared to its CPU counterpart (\S~\ref{sec:overallperformance}). The limitations of this design have been elucidated as challenges in~\S~\ref{introduction}.

\subsection{Solution Overview}
\label{solutionoverview}
% \vspace{-0.5mm}

To address these limitations, we introduce GPU-based biclique counting, named \textsf{GBC}, which includes a series of innovative designs. We initially formulate a search paradigm, specifically a hybrid DFS-BFS search, to enhance thread utilization (\S~\ref{sec:bicliquecounting}). Building upon this paradigm, we put forth a series of optimization techniques to further elevate the performance of \textsf{GBC}. Firstly, we introduce a novel data structure, i.e., Hierarchical Truncated Bitmap (HTB), featuring truncated bitmaps to facilitate fast intersection (\S~\ref{sec:bitmap}). We further devise a vertex reordering strategy tailored to compact bitmaps to mitigate memory overhead and access costs (\S\ref{sec:optimizations}-B). Secondly, we develop a joint load balancing strategy, which combines pre-runtime task allocation and runtime task stealing, to achieve an equitable distribution of workloads across GPU blocks (\S~\ref{sec:workloadbalance}). Lastly, we propose a communication-free graph partitioning method, known as \textsf{BCPar}, to handle graphs that surpass the capacity of device memory (\S~\ref{sec:outofmemory}).

\section{Hybrid DFS-BFS Exploration}
\label{sec:bicliquecounting}
% \vspace{-1mm}

In this section, we introduce the search paradigm in \textsf{GBC}, which serves as the backbone of our methodology.

The GPU baseline follows DFS exploration, which gives rise to two performance issues. 
Firstly, due to GPU memory transactions being executed in a coalesced manner, DFS exploration leads to bandwidth wastage as only one vertex is processed at a time. Secondly, as the search layers deepen, the sizes of intermediate results ($C_L[l]$ and $C_R[l]$) progressively decrease, often falling below the number of threads in a warp. 
Consequently, at deeper search levels, there are significantly fewer active threads compared to the large pool of available threads, resulting in substantially reduced thread utilization. 
% Conversely, while BFS presents substantial memory challenges, it is conducive to improving parallelism. Considering the strengths and weaknesses of both DFS and BFS, we harmonize their advantages while mitigating their drawbacks. This amalgamation results in a hybrid DFS-BFS strategy, where \textsf{GBC} operates in a DFS fashion globally while incorporating a local BFS strategy in specific regions. This approach enhances parallelism while effectively managing memory constraints. 

Note that, the search nodes at level $l$ sharing the same parent possess identical $C_L[l-1]$ and $C_R[l-1]$ (e.g., {$C_L$ and $C_R$ before update} in different nodes at level $l=1$ in Figure~\ref{basic_example}(b)). When applying intersection operations, the only distinction lies in their adjacency lists (e.g., $N_2^q(u_2), N(u_2)$ of $u_2$ vs. $N_2^q(u_3), N(u_3)$ of $u_3$ at level 1 in Figure~\ref{basic_example}). {Furthermore, the intersection operations among these children are mutually independent. }
{This observation presents an opportunity to simultaneously apply intersections for the children, akin to a BFS approach concerning the parent, which contributes to bandwidth utilization and enhanced parallelism.} 
However, embracing a global BFS exploration poses considerable memory challenges~\cite{lin2016network}. Consequently, considering the strengths and weaknesses of both DFS and BFS, we harmonize their advantages while mitigating their drawbacks. This assembly results in a hybrid DFS-BFS strategy, where \textsf{GBC} is executed in a DFS fashion globally while incorporating a local BFS strategy in specific regions. This approach enhances parallelism while effectively managing memory constraints. 
To clarify, let's consider a node in the search tree with $n$ children at level $l$. {In order to concurrently compute $C_L(l)$ for all children,} we duplicate $C_L[l-1]$ $n$ times and concatenate these duplicates to form an extended array, which is then stored in shared memory. Subsequently, each thread within a warp retrieves an element from this extended array and performs a binary search within $N_2^q(u)$. A similar approach is applied to calculate $C_R[l]$. 

Suppose there are $k$ warps handling tasks for the children of a specific vertex, and let $|C_L[l-1]| (|C_R[l-1]|) = m$. 
Without optimization, we would be required to perform $\lceil \frac{m}{32\times k} \rceil \times n$ intersection operations, while with optimization, only $\lceil \frac{m\times n}{32\times k} \rceil$ intersection operations are necessitated. 
In cases where $m < 32 \times k $, which is common in practice, the number of intersection operations is $n$ without optimization, compared to $\lceil \frac{m}{32\times k} \times n \rceil$ with optimization. It's worth noting that as the search level increases, the value of $m$ typically decreases, leading to a more efficient optimization performance.

% It's essential to note that at the first level, $S$ is initialized as the adjacency list of the visited vertex. In this case, the size of $S$ may not be small enough to make the optimized strategy significantly more efficient. Therefore, threads in this case directly fetch elements from the smaller of the two sets, either $S$ or $N(u,G)$, and search within the larger one to avoid unnecessary overhead.

\begin{figure}[tb]
    \centering
    \includegraphics[width=0.48\textwidth]{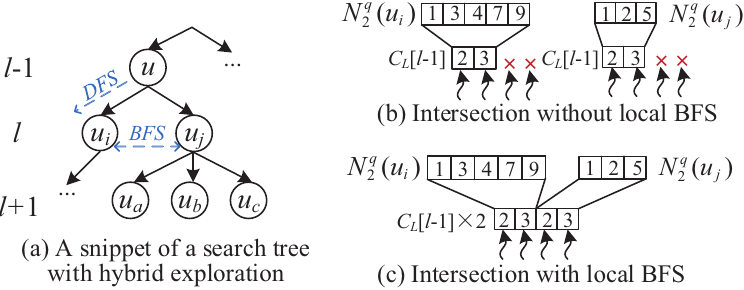}
    \vspace{-2mm}
    \caption{Hybrid DFS-BFS exploration in \textsf{GBC} (We use vertices newly added on the selected layer to denote nodes in the search tree for brevity).}
    
    \label{hybrid_exploration}
    \vspace{1mm}
\end{figure}

% \vspace{-2mm}
\begin{example} 
% \vspace{-1mm}
% \setlength{\abovedisplayskip}{1pt}
% \setlength{\belowdisplayskip}{1pt}
    Consider a snippet of a search tree in Figure~\ref{hybrid_exploration}(a), where we compute $C_L[l]$ for all children $\{ u_i,u_j\}$ of vertex $u$ at level $l$. Suppose there are 4 threads in a warp. As depicted in Figure~\ref{hybrid_exploration}(b), using only DFS with all threads assigned to a vertex ($u_i$ or $u_j$) to obtain $C_L[l]$ for ${ u_i, u_j}$ requires two separate intersection operations, with two out of four threads remaining idle, resulting in a significant waste of computing resources. In contrast, with hybrid DFS-BFS exploration, where we handle $u_i$ and $u_j$ concurrently by duplicating $C_L[l-1]$, only one intersection is necessary to acquire the results for ${ u_i, u_j}$, as depicted in Figure~\ref{hybrid_exploration}(c). This approach leads to a more efficient utilization of computing resources, ultimately reducing time consumption.
% \vspace{-0.5mm}
\end{example}
% \vspace{-5mm}

\setlength{\textfloatsep}{0pt}
\begin{algorithm}[t]
    \caption{\textsf{GBC}}
    \label{gbc}
    \LinesNumbered
    \linespread{0.9}\selectfont
    \KwIn{a bipartite graph $G$, two integer $p$ and $q$}
    \KwOut{all $(p, q)$-bicliques $\mathcal{B}$}
    % \SetKwProg{Fn}{procedure}{}{}
    \SetKwFunction{GPUBasedListing}{$\mathsf{GPUBasedListing}$}
    \SetKwFunction{IntersectionBatch}{$\mathsf{IntersectionBatch}$}
    % Transform $H$ to a DAG;\\
    Select one layer in $G$ as an anchor\;
    Collect 2-hop neighbors $N_2^q(\cdot)$ for vertices in the anchored layer\;
    Filter unpromising vertices and obtain a set $Roots$\;
    \GPUBasedListing{$G, \mathcal{B}, p, q, Roots$}\;
    \Return{$\mathcal{B}$}\; %\linebreak[1]

    \textbf{procedure} \GPUBasedListing{$G,\mathcal{B}, p, q, Roots$}\\
    $i \leftarrow blockIdx$\;
    \While{$i < |Roots|$}{
        $C_R[blockIdx], C_L[blockIdx]\leftarrow p$ empty arrays\;
        % $H\_Sub[blockIdx]\leftarrow p$ arrays initialized as empty\;
        $u\leftarrow Roots[i]$, $L\leftarrow \emptyset$, $l\leftarrow 0$\;
        $C_R[blockIdx][0]\leftarrow N(u)$\;
        $C_L[blockIdx][0]\leftarrow N_2^q(u)$\;
        \While{$l\ge 0$}{
            $B\leftarrow$ next batch in $C_L[blockIdx][l]$\;
            \uIf{$threadIdx < set\_num$}{
                \IntersectionBatch{$B, C_R, l, V$}\;
            }
            \Else{
                \IntersectionBatch{$B, C_L, l, U$}\;
            }
            \If{$l = p$}{
                \ForEach{$R\subseteq C_R[l-1] \wedge |R|=q$}{
                    $\mathcal{B}\leftarrow \mathcal{B}\cup {(L,R)}$\;
                }
                $l \leftarrow l - 1$, $L\leftarrow L-\{u\}$\;
                \Return\;
            }
            \ForEach{$u^\prime\in C_L[blockIdx][l]$}{
                % $|S[blockIdx][l][v':v'+1]|\ge q \wedge$ \\ $|H\_Sub[blockIdx][l][v':v'+1]|\ge p - l - 1$
                \If{Pruning conditions are not satisfied}{
                    $u \leftarrow u^\prime$, $l\leftarrow l + 1$, $L\leftarrow L\cup \{u^\prime\}$\;
                    $C_L[blockIdx][l] - \{u^\prime\}$\;
                }
            }
            \If{$C_L[blockIdx][l] = \emptyset$}{
                $l\leftarrow l - 1$, $L\leftarrow L-\{u\}$\;
            }
        }
        $i \leftarrow i + gridDim$\;
    }
    \textbf{procedure} \IntersectionBatch{$B, Arr, l, G$}\\
    \While{$B \ne \emptyset$}{
        Copy $Arr[l-1]$ $|B|$ times into buffer\;
        Intersect with $N(u^\prime)$ or $N_2^q(u^\prime)$ where $u^\prime\in B$\;
        Write results into $Arr[l]$\;
    }
\end{algorithm}
\setlength{\textfloatsep}{0pt plus 0pt minus 0pt}

% {\textbf{Vertex Ordering.} To avoid duplicate results, we adopt an accessing order from the low degree vertices to the high degree vertices. Additionally, by doing this, the workload of each thread block will be more balanced due to the power law distributation of vertex degrees. } 

% \vspace{1mm}
\textbf{Batching.} The shared memory's capacity is quite limited and adhering to the BFS phase of the previous hybrid exploration strategy may lead to an explosion of shared memory usage with some large-sized adjacency lists.
To overcome this limitation and ensure the strategy aligns with the constraints of limited memory, we resort to batching. This adaptation allows us to effectively manage the memory while maintaining the benefits of parallelism. 
Concretely, \textsf{GBC} divides the children into multiple independent batches with size $\lfloor \frac{|B|}{|C_L[l-1]|} \rfloor$, where $B$ is the buffer allocated in shared memory for storing the duplicates. Once a batch is processed using BFS, it switches to DFS instead of continuing with the remaining batches.

% \vspace{-2mm}
Algorithm~\ref{gbc} shows the pseudo-code of \textsf{GBC}. 
At first, we select one layer and collect the 2-hop neighbors for each vertex in the selected layer\skip\footins\smallskipamount\footnote{We adopt the layer selecting strategy proposed in~\cite{yang2021p} for its effectiveness.}. 
The main part of the algorithm is \textsf{GPUBasedListing}. In this procedure, we first initialize the arrays $C_R$ and $C_L$ to store the intermediate results and distribute root vertices to each thread block (lines 9--12). 
Then we retrieve bicliques level by level (lines 13--29). For each batch $B$ in $C_L[blockIdx][l]$, we employ the first $set\_num$ threads to calculate $C_R[l]$ by intersecting $C_R[l-1]$ and $N(u^\prime)$ for $u^\prime \in B$ via \textsf{IntersectionBatch}. While the other threads calculate the results of $C_L[l]$ simultaneously (lines 14--18). 
When the search reaches the last level, we store the qualified bicliques into $\mathcal{B}$ and backtrack (lines 19--23). Otherwise, for each vertex $u^\prime\in C_L[blockIdx][l]$, we examine if the intersection results satisfy pruning conditions, i.e. $|C_R\cap N(u^\prime)| \ge q$ and $|C_L\cap N_2^q(u^\prime)| \ge p - l - 1$. 
The qualified vertices move to the next iteration (lines 24--27). 

% \vspace{-2.5mm}
\noindent {\textbf{Discussion.}}
{ 
% We compare \textsf{BCLP}~\cite{yang2021p} with \textsf{GBC}. 
\textsf{GBC} adopt the layer-based approach as \textsf{BCLP}~\cite{yang2021p} for efficiency consideration. However, they differ significantly.
% However, they target different platforms and utilize distinct methodologies.
\underline{First,} \textsf{BCLP} adopts iteration instead of recursion to alleviate memory issues and opens the potential for reusing intermediate results.
% \textsf{BCLP} relies on recursion on the CPU, whereas \textsf{GBC} adopts iteration on GPU, due to the memory management challenges, high memory consumption and difficulty in reusing intermediate results inherent in recursion.
Moreover, to fully harness the parallel power of the GPU, we implement a hybrid DFS-BFS exploration rather than pure DFS as in \textsf{BCLP}. 
\underline{Second,} \textsf{BCLP} preallocates arrays for each vertex (i.e., task) on the selected layer, leading to high memory demand on the GPU due to the large number of parallel tasks. Additionally, the labeling technique of \textsf{BCLP} introduces substantial data movement, which proves time-consuming on the GPU. 
Therefore, \textsf{GBC} employs parallel intersection computation, which is efficient on the GPU and helps conserve memory.
% \textsf{BCLP} utilizes a subgraph-based method for retrieving sub-cliques on the selected layer. However, this approach entails building a subgraph for each vertex (task) on the selected layer, leading to high memory demand on the GPU with its thousands of parallel tasks. Additionally, the labeling technique of \textsf{BCLP} introduces substantial data movement, which is time-consuming on the GPU. Therefore, \textsf{GBC} employs parallel intersection computation, which is efficient on the GPU and helps save memory.
\underline{Third,} we introduce novel optimizations for enhanced efficiency, including optimizing data structures, vertex reordering, and load balancing strategies. 
In summary, we did not directly utilize the parallel framework of \textsf{BCLP}. Instead, we implemented numerous non-trivial optimizations specifically tailored for GPU architecture to achieve superior performance with \textsf{GBC}.}
% we propose further optimizations to enhance the efficiency of the parallel framework, including optimizing data structures, vertex reordering, and load balancing strategies. In summary, we did not directly utilize the parallel framework of \textsf{BCLP}. Instead, we implemented numerous non-trivial optimizations specifically tailored for GPU architecture.}

% \vspace{-5mm}
\section{Optimizations}
\label{sec:optimizations}
% \vspace{-3.5mm}

\begin{figure*}[tbp]
\centering
    \includegraphics[width=0.92\textwidth]{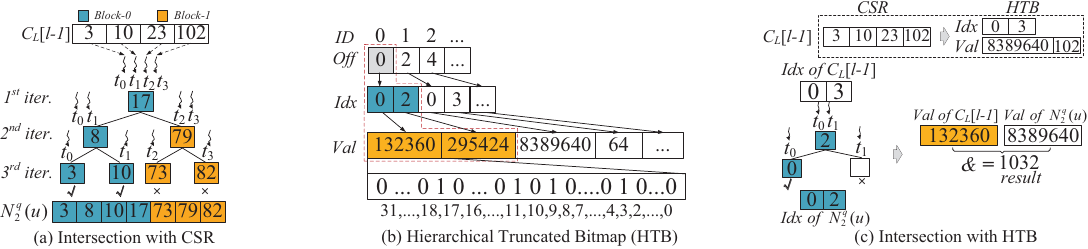}
    \vspace{-2mm}
    \caption{Intersection with different data structures.}
    \label{fig:htb}
    \vspace{-7mm}
\end{figure*}

% In this section, we propose several advanced techniques to optimize the performance of \textsf{GBC}. 
{In this section, we first discuss the limitations of parallel binary search for intersection computation. Motivated by this, we present an advanced data structure and vertex reordering technique to enhance intersection implementation.
Finally, we address load imbalance through pre-runtime and runtime workload distribution.}

% -------------------------------------------------------------Hierarchical Truncated Bitmap---------------------------------------------------------
% \vspace{-2mm}
\subsection{Hierarchical Truncated Bitmap}
\label{sec:bitmap}
% \vspace{-3mm}

{As mentioned in \S~\ref{introduction}, intersection computation constitutes the majority of the workload, accounting for over 90\%.}
Although proven to be efficient~\cite{Ao2011}, two drawbacks persist in parallel binary search on GPU. 
Firstly, binary search necessitates entry-by-entry comparisons, resulting in high computational overhead.
Secondly, comparisons in each iteration require accessing entries stored in global memory, potentially leading to excessive memory transactions. 
{The following example illustrates how intersection performs with binary search under CSR format on GPU and its corresponding limitations.} 

% \vspace{-1.75mm}
\begin{example}
    In Figure~\ref{fig:htb}(a), suppose we compute $C_L[l]=C_L[l-1]\cap N_2^q(u)$ with $C_L[l-1] = \{3,10,23,102\}$ and $N_2^q(u) = \{3,8,10,17,73,79,82\}$, and the entries in $C_L[l-1]$ work as the search keys (same workflow for $C_R[l]$). For simplicity, we assume there are 4 threads in a warp, and each memory transaction accesses 4 integers. 
    In the first iteration, 4 threads access entry 17 residing in the same block with one memory transaction. 
    In the second iteration, thread-0,1 and thread-2,3 access entries 8 and 79 distributed across two data blocks, hence, two memory transactions are required. 
    Similarly, another two memory transactions are required. 
    % Totally, binary search, in this case, performs five memory transactions, while a majority of them are subject to repeated scheduling. Specifically, block-0 and block-1 have been transmitted thrice and twice, respectively.     
\label{example:bs}
% \vspace{-2mm}
\end{example}
% \vspace{-1.75mm}

\noindent {\textbf{Limitation.}
Totally, binary search, in Example~\ref{example:bs}, performs five memory transactions, while a majority experiencing \textit{repeated scheduling}. Specifically, block-0 and block-1 have been transmitted thrice and twice, respectively.
These substantial memory transactions stem from the \textit{entry-by-entry comparisons} with entries residing in different blocks of global memory.} 

% To narrow down the search space and reduce the transaction overhead, we proposed a new GPU-friendly data structure, namely Hierarchical Truncated Bitmap (HTB), to store the adjacency lists. 

{Inspired by the observations above, the key insight to enhance the efficiency of parallel binary search is to narrow down the search space to avoid entry-entry comparisons and compress the data to reduce transaction overhead.
To achieve this, we devise a new GPU-friendly data structure, known as Hierarchical Truncated Bitmap (HTB), to store the adjacency lists.
The main idea of HTB is to use a single 32-bit integer to represent multiple neighbors of a vertex, where each integer accommodating neighbor IDs within the range $[i, i+32]$, where $i$ signifies the ordinal position in the consecutive integer sequence.} 
% Given that an integer can accommodate 32 bits, HTB regards it as a collection of non-negative integers representing neighbor IDs within the range $[i, i+32]$, where $i$ signifies the ordinal position in the consecutive integer sequence. 
For each neighbor $u^\prime$ of vertex $u$, we hash it to the $j$-th bit in the $i$-th integer and set the bit value to 1 to denote its presence, where $i = ID(u^\prime)/32$, $j = ID(u^\prime)\%32$. We then aggregate $i$s and perceive it as an index for these collections. 

% Figure~\ref{fig:htb}(b) illustrates an intuitive sketch of HTB, which has a similar structure to the widely adopted CSR format. 
% HTB consists of three tiers of arrays: 
% (i) \textbf{\textit{Off}} shares similar functionality with the row index in CSR, and we use it to denote the starting position of $Idx$ corresponding to each vertex. Hence, the index of vertex $u$ lies in $Idx[Off[ID(u)]:Off[ID(u)+1]-1]$, 
% (ii) \textbf{\textit{Idx}} is employed to signify the orders of the integers that store the neighbors, i.e., $i$s, 
% and (iii) \textbf{\textit{Val}} is leveraged to indicate the specific bit position within a 32-bit integer where a neighbor is stored, i.e., $j$s.

{Figure~\ref{fig:htb}(b) presents a sketch of HTB, which consists of three tiers of arrays: 
(i) \textbf{\textit{Off}} denotes the starting position of $Idx$ corresponding to each vertex, wherein the index of vertex $u$ occupies $Idx[Off[ID(u)]:Off[ID(u)+1]-1]$, 
(ii) \textbf{\textit{Idx}} signifies the series of the integers that conserve the neighbors, i.e., $i$s, 
and (iii) \textbf{\textit{Val}} indicates the specific bit position within a 32-bit integer where a neighbor is stored, i.e., $j$s.}

% \vspace{-2mm}
\begin{example}
% \vspace{-0.2cm}
    Suppose \textit{ID}$(u)=0$, for each $u^\prime\in N_2^q(u)$ in Figure~\ref{fig:htb}(a), HTB hashes them into $(i,j)$ pairs, which are $\{ (0,3),(0,8),(0,10),(0,17),(2,9),(2,15),(2,18) \}$. 
    For instance, $u^\prime$ with $ID(u^\prime)=3$, whose $(i,j)$ pair is $(0,3)$, is stored in the $0$-th integer with the $j$-th bit being 1. 
    Consequently, $N_2^q(u)$ is distributed into two integers, i.e., the $0$-th and $2$-nd integer stored in $Idx[Off[0]:Off[1]-1] = \{ 0,2\}$, with values in $Val$ being $\{132360,295424\}$, respectively.
% \vspace{-0.2cm}
\end{example}
% \vspace{-1.5mm}

Intersection is executed in a two-phase manner with HTB. In the first phase, we leverage $Idx$ to determine the search range, followed by the extraction of $Val$ containing the collection of neighbors. 
In the second phase, a bitwise AND operation is conducted to ascertain whether the search key exists in the search list or not. 
Figure~\ref{fig:htb}(c) depicts the workflow of binary search with HTB.

% \vspace{-2mm}
\begin{example}
\label{htb_binary}
    Utilizing HTB, we first perform a binary search with the $Idx$ of $C_L[l-1]$ ($\{0,3\}$) over that of $u$ ($\{ 0,2 \}$), yielding the search result $\{ 0 \}$ (left part in Figure~\ref{fig:htb}(c)). 
    Thereafter, we circumvent the need for searching entries in $N_2^q(u)$ stored in the $3$-rd integer. 
    Next, we perform the bitwise AND operation between the values in $Val$ of $C_L[l-1]$ and $N_2^q(u)$ corresponding to the $0$-th integer (i.e., 8389640 and 132360), resulting in the outcome $1032$, including neighbors $\{3,10\}$ (right part in Figure~\ref{fig:htb}(c)). 
\end{example}
% \vspace{-2mm}

We analyze the number of memory transactions in Example~\ref{htb_binary}. 
Executing binary search requires one memory transaction to load the $Idx$ of $N_2^q(u)$ and another memory transaction to access the corresponding entries in $Val$. Consequently, two memory transactions are needed, leading to a reduction of three memory transactions compared to directly executing a binary search with CSR.
Furthermore, the utilization of bitwise AND operations contributes to computational efficiency.

% -------------------------------------------------------------Vertex Reordering---------------------------------------------------------
% \vspace{-4mm}
\subsection{Vertex Reordering}
\label{sec:vertexreorder}
% \vspace{-2mm}

HTB's efficiency depends on the degree of discreteness in the adjacency lists. 
{The narrower the gap between adjacent entries, the fewer integers are required to accommodate the adjacency lists. This results in higher efficiency for HTB, as fewer memory reads are required and more entries can be compared in a single bitwise AND operation.
% In the cases where vertices in an adjacency list are hashed to integers whose quantity approximately matches the list's size, HTB's efficiency would lag behind CSR.
In this section, we introduce \textsf{Border}, a vertex reordering technique to optimize the data layout to reduce the size of HTB.}

Various methods have been devised for reordering vertices in unipartite graphs to improve the hit rate of cache~\cite{sha2021self,boldi2011layered,arai2016rabbit,WeiYLL16}.
% such as \textsf{LLP}~\cite{boldi2011layered}, \textsf{Rabbit Reordering}~\cite{arai2016rabbit}}, while \textsf{Gorder} for improving hit rate of CPU cache.
% Existing vertex reordering methods, such as \textsf{Gorder}~\cite{WeiYLL16}, are designed for unipartite graphs\textcolor{red}{[LS: the purpose of Gorder]} {to improve the hit rate of CPU cache}.
However, reordering vertices for $(p,q)$-biclique counting on GPU poses new challenges. 
First of all, we have to reorder vertices in different layers of $G$ separately since reordering them as a whole would disrupt the vertices' orders on the left and right sides. More importantly, HTB introduces a new requirement for the reordering algorithm to maximize the number of vertices hashed into one integer.
Motivated by the above considerations, we devise \textsf{Border} tailored to reorder vertices of the bipartite graph to prevent their disruption while optimizing the layout for efficiency consideration of HTB. 

% However, for $(p,q)$-biclique counting, we face the challenge of reordering vertices in different layers of $G$ separately. We cannot reorder all the vertices in the whole graph of $G$, as doing so would disrupt the order of vertices on the left and right sides of the bipartite graph\textcolor{red}{[LS: how about fixing one layer when reordering another layer?]}. What's more, HTB introduces a new requirement for the reordering algorithm, namely, the need to maximize the number of vertices hashed into one integer\textcolor{red}{[LS: the key of Border?]} {which is different from the purpose of the existing reordering algorithms}. 
% Motivated by these considerations, we devise \textsf{Border} tailored to reorder the vertices within one layer of $G$ {(different layers of $G$ separately?)}\textcolor{red}{[LS: only one layer?]} while optimizing the layout for efficiency consideration of HTB. 

Algorithm~\ref{Rorder} shows the pseudo-code of \textsf{Border}. 
\textsf{Border} first maps $G$ to an adjacency matrix $M$, where $M(i,j)\in \{0,1\}$ denotes the entry located at the $i$-th row and the $j$-th column. Within this matrix, we refer to 32 entries whose indices are within $[i, 32\times k: 32\times (k+1)](k \in \mathbb{Z}^{+})$ in each row as a $block$. An $m$-$block$ refers to \textit{block} containing $m$ 1s. 
Our observation on real datasets reveals that the abundance of 1-\textit{block}s significantly dominates the realm of non-zero \textit{block}s, which deteriorates the efficiency of HTB owing to their sparsity. 
Consequently, the pursuit of minimizing the quantity of 1-\textit{blocks} results in a more condensed bit representation within an integer of HTB, thus reducing the sizes of \textit{Val} required to encode each adjacency list and improve HTB's efficiency.
\textsf{Border} employs a greedy strategy in each iteration to minimize the number of $1$-$block$s, which includes four steps: 

% \vspace{-1mm}
\begin{itemize}
    \item \textbf{\textit{Step 1}}: Identify vertex $v_{m}$ with most $1$-$block$s (line 9); 
    \item \textbf{\textit{Step 2}}: Construct the candidate set $CandV$ by including vertices sharing the fewest common neighbors with $v_{m}$ (lines 10--12); 
    \item \textbf{\textit{Step 3}}: Calculate profits for all vertices in $CandV$ and select vertex $v_{n}$ with the highest profit (lines 13--16); 
    \item \textbf{\textit{Step 4}}: Exchange the positions of $v_{m}$ and $v_{n}$ and update matrix $M$ (lines 17--19). 
\end{itemize}
% \vspace{-1mm}

Intuitively, each iteration endeavors to reduce the quantity of $1$-$block$s in $M$ as much as possible, and the profit derived from exchanging $v_m$ and $v_n$ is numerically equivalent to the count of reduced $1$-$block$s.
Concerning the $block$s that $v_m$ locates at, the change in the number of $1$-$block$s is divided into two parts: the number of $1$-$block$s converted into $0$-$block$s, denoted as $x_m$, and the number of $0$-$block$s converted into $1$-$block$s, denoted as $y_m$. 
Similarly, for $v_n$, these two parts are denoted as $x_n$ and $y_n$. 
The profit is calculated as $x_m+x_n-y_m-y_n$. 

Notably, we reorder the vertices based on their degrees before executing \textsf{Border}. 
{This preprocessing enhances the compactness of the layout for the adjacency lists of each vertex, consequently diminishing the count of $1$-$block$s and thereby reducing the number of iterations.}
Additionally, step 2 in \textsf{Border} is efficiently performed using matrix multiplication. 

\setlength{\textfloatsep}{0pt}
\begin{algorithm}[t]
    \caption{\textsf{Border}}
    \label{Rorder}
    \LinesNumbered
    \linespread{0.9}\selectfont
    \KwIn{a bipartite graph $G(U,V,E)$, number of iterations $itr$, reorder layer $rl$}
    \KwOut{$G$ with reordered vertices}
    
    \textbf{if} $flag=0$ \textbf{then} exchange $U, V$\;
    Map $G$ to an adjacency matrix $M$\;
    $OBV\leftarrow \emptyset$, $CandV\leftarrow \emptyset$, $max\_profit\leftarrow 0$\;
    \ForEach{$blk\in 1$-$block$s}{
        \If{$blk[\cdot,v] = 1$}{
            $OBV$.append($v$)\;
        }
    }
    % Add the vertices located at $1$-$block$s to $OneBlockVertices$\;
    Count the number of $1$-$block$s for each vertex in $OBV$\;
    \While{$itr--$}{
        Find the vertices $v_m$ with the most $1$-$block$s\;
        \ForEach{$v\in U$}{
            \If{$|N(v) \cap N(v_m)|$ is the smallest}{
                $CandV \leftarrow CandV \cup \{v\}$\;
            }
        }
        \ForEach{$v \in CandV$}{
            \If{$Profit(v) \ge max\_profit$}{
                $v_n\leftarrow v$\;
                $max\_profit \leftarrow Profit(v)$\;
            }
        }
        Exchange the order of $v_m$ and $v_n$\;
        Update the number of $1$-$block$ of each vertex\;
        Update the matrix $M$\;
    }
    \Return{$G$}\;
\end{algorithm}
\setlength{\textfloatsep}{0pt plus 0pt minus 0pt}

% \vspace{-1mm}
\begin{example} 
% \vspace{-2mm}
    Using the adjacency matrix in Figure~\ref{reorderexample}(a) as an example to illustrate an iteration in Algorithm~\ref{Rorder}, assuming 4 bits constitute a block of HTB for brevity.
    
    In Figure~\ref{reorderexample}(b), the integers positioned above the vertices indicate the quantity of $1$-$block$s (depicted as gray boxes) for vertices to be reordered.
    \textsf{Border} first identifies the vertex with the highest count of $1$-$block$s, specifically, v$_5$.
    Subsequently, \textsf{Border} finds vertices sharing the fewest common neighbors with $v_5$ as candidates to exchange order. This is achieved by multiplying the vector $(1,1,1,1)$ (neighbors of $v_5$) with matrix $M$, resulting in $CandV = \{v_1,v_3,v_4,v_{11}\}$.
    \textsf{Border} proceeds to compute the potential profits, denoted by the integers above the vertices in Figure~\ref{reorderexample}(c), associated with the exchange of $v_5$ with candidates in $CandV$. 
    \textsf{Border} finally selects $v_3$ for order exchange with $v_5$, which yields the maximum profit. 
    The matrix after the exchange is shown in Figure~\ref{reorderexample}(d). 
\end{example}
% \vspace{4pt}

% -------------------------------------------------------------Load Balancing---------------------------------------------------------
% \vspace{-18pt}
\subsection{Load Balancing}
\label{sec:workloadbalance}
% \vspace{-6pt}

% The disparate sizes of adjacency lists lead to imbalanced workloads. Additionally, the dynamic nature of DFS-based exploration poses challenges in estimating workloads during runtime, exacerbating the load imbalance issue. 
{Load imbalance is a common dilemma in multi-threaded environments. In the studied problem, two factors contribute to this issue:
(1) the disparate sizes of adjacency lists lead to uneven workloads, and (2) the dynamic nature of DFS-oriented exploration poses challenges in estimating workloads during runtime, exacerbating the load imbalance issue and hindering the application of only static solutions.}
To tackle this issue, we formulate a joint load balancing strategy that integrates pre-runtime task allocation with runtime task stealing. 
% These strategies work in tandem to effectively address load imbalance.

\textbf{Pre-runtime load balancing.}
Evenly distributing vertices in the selected layer among different blocks results in pronounced load imbalance due to variations in the sizes of search trees that arise from the uneven lengths of adjacency lists. 
It is noteworthy that the height of each search tree remains consistent, with its value being $p$ or $q$, depending on the selected layer. Consequently, allocating vertices located at higher levels in the search trees to each thread block yields a more equitable distribution of workloads among the blocks. 
We embrace an edge-oriented approach that evenly distributes the second-level vertices of the search trees across different thread blocks. 
% The experimental results presented in~\S~\ref{sec:evaluation} demonstrate the effectiveness of this strategy. 

\begin{figure}[tbp]
    \centering
    % \vspace{-2mm}
    \includegraphics[width=0.49\textwidth]{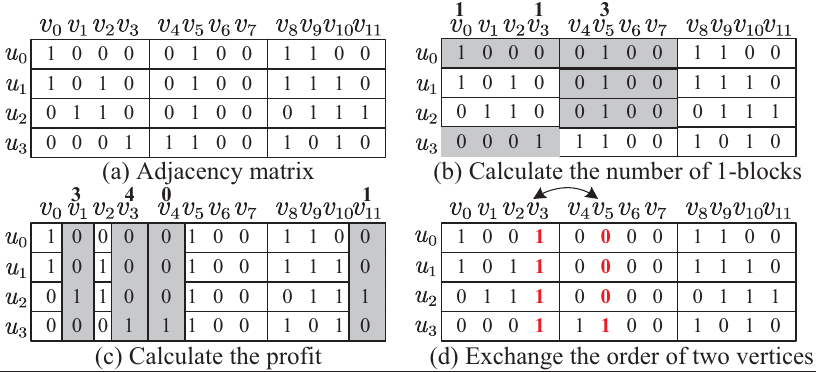}
    \vspace{-6mm}
    \caption{A running example of \textsf{Border}}
    \label{reorderexample}
    \vspace{1mm}
\end{figure}

\textbf{Runtime load balancing.}
The dynamic nature of DFS, characterized by unquantifiable workloads, gives rise to disparities of workloads among the blocks at runtime. 
To guarantee a more equitable distribution of workloads, we implement a work stealing mechanism, whereby the unoccupied blocks are dynamically assigned (sub)tasks from those currently engaged. 

% \textit{GCL} is an array of the same size as the number of blocks on the GPU, where \textit{GCL}[\textit{i}] records the vertex ID being processed by block-$i$ and serves as the starting point for work stealing. \textit{GCL}[$i$] is set to \textit{0xFFFFFFFF} to indicate that all vertices assigned to block-$i$ have been processed, where work stealing should skip.

% In Figure~\ref{worksteal}, we maintain an array \textit{GCL} with the size of the number of blocks on the GPU to denote the processing status of blocks, where \textit{GCL}[\textit{i}] records the number of vertices have been processed by block-$i$ and serves as the starting point for work stealing. Once a block concludes its assigned tasks, the corresponding \textit{GCL} value is set to \textit{0xFFFFFFFF}, initiating the work stealing process.
{In Figure~\ref{worksteal}, we maintain an array \textit{GCL} with the same size as the number of blocks on the GPU  to denote the processing status of blocks, where \textit{GCL}[\textit{i}] records the number of processed vertices by block-$i$ and serve as the starting point if the work is stolen. \textit{GCL}[$i$] will be set to \textit{0xFFFFFFFF} when all vertices assigned to block-$i$ have been processed, indicating that work stealing should skip this block.
% An intuitive way of selecting a block to be stolen is to identify the slowest block, namely the block with the minimum \textit{GCL} value.} However, traversing the entire \textit{GCL} proves time-consuming in scenarios with a substantial number of blocks. Furthermore, both acquiring work from the slowest block and stealing work from other blocks yield the same effect eventually. 
We have the block that has completed its assigned tasks persist in searching for unfinished blocks through \textit{GCL}, as shown in Figure~\ref{worksteal}(a). 
Once a qualified block is identified, as depicted in Figure~\ref{worksteal}(b), it locks the corresponding entry (block) in \textit{GCL} and reads the value. 
Subsequently, it calculates the index of the next vertices to be processed and updates the \textit{GCL} array as in Figure~\ref{worksteal}(c).
Finally, it frees the locked entry as illustrated in Figure~\ref{worksteal}(d).}
% However, employing a lock on the entire \textit{GCL} each time a block accesses it would significantly curtail parallelism. As there is no conflict when two blocks access \textit{GCL} values of distinct blocks, we establish an independent lock for each value in \textit{GCL}, corresponding to each block. This approach results in a more efficient memory access mechanism.

% merrill2012scalabl——Scalable GPU graph traversal, D Merrill, M Garland, A Grimshaw
% wang2016gunrock —— gunrock
% wang2019sep —— sep-graph

Note that,~\cite{merrill2012scalable,wang2016gunrock,wang2019sep} statically allocate thread groups based on adjacency list lengths, which cannot ensure load balance due to the dynamic nature of computations~\cite{kumar2019issues}. 
Additionally,~\cite{wang2016gunrock} incurs overhead from adjacency list partitioning, while~\cite{merrill2012scalable} requires synchronization before thread allocation. 
In contrast, we integrate dynamic work stealing to ensure load balance during execution by efficiently utilizing idle threads.
% For instance, Gunrock divides the adjacency list into equally sized chunks for processing across blocks. The other two methods allocate threads, warps, and thread blocks based on the adjacency list size. 
% They still cannot guarantee load balance due to the dynamic nature of the computations~\cite{kumar2019issues}. 
% In contrast, we integrate dynamic work stealing to ensure load balance during execution.
% Additionally,~\cite{wang2016gunrock} incurs overhead from adjacency list partitioning, while~\cite{merrill2012scalable} requires synchronization before thread group allocation. In contrast, our method efficiently utilizes idle threads for work stealing.
% Furthermore, these methods incur additional overhead. For instance, Gunrock requires adjacency list partitioning, while Merrill et al's method necessitates synchronization before thread group allocation. In contrast, our method efficiently utilizes idle threads for work stealing.

\begin{figure}[htbp]
    \centering
    \vspace{-4mm}
    \includegraphics[width=0.43\textwidth]{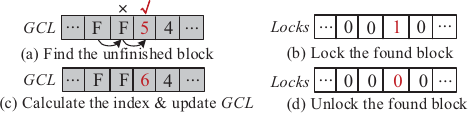}
    \vspace{-2mm}
    \caption{{Process of work stealing (F: \textit{0xFFFFFFFF}, red text: identified block)}}
    \label{worksteal}
    \vspace{-4mm}
\end{figure}

% \input{Scalability}
% \vspace{-1.5mm}
\section{Out-of-memory Setting}
\label{sec:outofmemory}
% \vspace{-1.5mm}

Given that global memory on GPU is considerably smaller than its CPU counterpart, typically ranging from a few gigabytes to several tens of gigabytes, the feasibility of holding the entire graph in global memory has been rendered impractical due to the swift proliferation of real-world graphs. 
With their sizes surpassing several tens of gigabytes and, in some cases, even reaching several hundred gigabytes, the necessity of segmenting large graphs into more manageable subgraphs that can fit within the confines of global memory has become a pressing imperative.
Despite certain literature having delved into graph partitioning~\cite{Karypis1998,chen2019powerlyra,andreev2004balanced}, their oversight in considering the properties of the bicliques has led to superfluous data transfer, which even dominates the overall performance. This occurs when required data is transferred on demand, and a certain portion of data is transferred multiple times~\cite{guo2020gpu}. 
To enhance graph partitioning for biclique counting on GPU, we introduce a biclique-aware partitioning solution named \textsf{BCPar}.

% {Reconsider the search for $(p,q)$-bicliques starting from $u$.} 
{We observe that since $C_L[l]$ is obtained through iterative intersections with $C_L[0]$, i.e., $N_2^q(u)$, the vertices in $C_L[l]$ necessarily belong to $N_2^q(u)$, i.e., $\forall u^\prime \in C_L[l], u^\prime \in N_2^q(u), 0\leq l \leq p-1$.}
% , the combination of $N_2^q(u)$. 
To construct $C_L[l]$, we intersect $C_L[l-1]$ with the 2-hop neighbors $N_2^q(u^\prime)$ of the newly added vertex $u^\prime$ in $L$. 
Hence, we only need to consider the 2-hop neighbors $N_2^q(u)$ of the starting vertex $u$ and 2-hop neighbors $N_2^q(u^\prime)$ of $u$'s 2-hop neighbors $u^\prime\in N_2^q(u)$ for computing $C_L$ of $u$. Similarly, $C_R$ is determined by the 1-hop neighbors of the starting vertex and its 2-hop neighbors.

Thereafter, to retrieve bicliques for $u$, we collectively group the 1-hop and 2-hop neighbors of $\{u\cup N_2^q(u)\}$, rather than loading them on demand.
{As the neighbors of a vertex may be shared by others, we aim for maximal sharing of the neighbors among vertices within the same partition to minimize partition sizes and the number of partitions.}
However, achieving the optimal partitioning result is NP-hard (Lemma~\ref{lemma:partition}). We propose a greedy solution, namely \textsf{BCPar}, to address the partition problem, as shown in Algorithm~\ref{alg:bcpartitioner}.

% Reconsider the search procedure of a $(p,q)$-biclique $B$ starting from vertex $u$, which we search by finding a $p$-clique $C$ in $H$ (suppose $H$ constructed on layer $L$) along with the intersection of each $N(u^\prime, G), u^\prime \in C$. Since $p$-clique $C$ is a complete graph containing $u$, $C$ must be the combination of $u$ and it's $(p-1)$ $1$-hop neighbors. Hence, we only have to consider at most $2$-hop neighbors of $u$ in $H$ to search a biclique, i.e., $1$-hop neighbors as the candidates to form a clique and $2$-hop neighbors to filter candidates. A straightforward idea is to load both the 1-hop and 2-hop neighbors of a vertex into device memory before calculating its biclique, rather than loading them on demand. Since a vertex and its neighbors may be shared by other vertices, we hope that the vertex and its neighbors can be shared by other vertices in the same partition as much as possible to reduce the number of partitions as well as data loading time. However, achieving the optimal partitioning result is NP-hard(Lemma~\ref{lemma:partition}). We propose a greedy solution to solve the partition problem, as shown in Algorithm~\ref{alg:bcpartitioner}.

% \vspace{-2mm}
\begin{lemma}
\label{lemma:partition}
    Obtaining the optimal graph partitioning result is an NP-hard problem.    
\end{lemma}
% \vspace{-4mm}

% \vspace{-1mm}
\begin{proof}
\label{proof:partition}
    The proof of Lemma~\ref{lemma:partition} is straightforward with the knapsack problem and thus is omitted.
\end{proof}
% \vspace{-2mm}

\setlength{\textfloatsep}{0pt}
\begin{algorithm}[!t]
\caption{\textsf{BCPar}}
\label{alg:bcpartitioner}
\LinesNumbered
\linespread{0.9}\selectfont
\KwIn{vertex set $U$ of selected layer, neighbor structures $N(\cdot)$ and $N_2^q(\cdot)$, memory budget $M$ }
\KwOut{partition $P = \{ g_1,g_2,...,g_k \}$}

Compute weight $w(u) = |N(u)| + |N_2^q(u)|$, $u\in U$\;
Compute average weight $avg_w(u) = \frac{1}{|N_2^q(u)|} \sum_{u^\prime \in N_2^q(u)} w(u^\prime)$, $u\in U$\;
Initialize an array ${L}$ with $U$ in descending order\;
% Initialize priority query $\mathcal{Q}$ with $V(H)$\;
$P \leftarrow \emptyset$\;
\While{$U \neq \emptyset$}{
    $p \leftarrow \emptyset$, $p^\prime \leftarrow \emptyset$, ${Q} \leftarrow \emptyset$, $cost \leftarrow 0$\;
    $u \leftarrow {L}$.pop(), $p$.insert($u$)\;
    
     \ForEach{$u^\prime \in N_2^q(u)\cup u$}{
        $p^\prime$.insert($u^\prime$)\;
        $cost+=(|N(u^\prime)|+|N_2^q(u^\prime)|)$\;
     }
     
    \While{ True }{
        \ForEach{$u^\prime \in p^\prime$}{
            \ForEach{$v \in N_{in} (u^\prime$)}{
                \If{$v \notin {Q}$}{${Q}$.insert($v$)\;}
                Increase the weight of $v$ in ${Q}$ by $w(v)$\;
            }
        }
        $u \leftarrow {Q}$.pop()\;
        \ForEach{$u^\prime \in N_2^q(u)\cup u$}{
            \If{$u^\prime \notin p^\prime$}{
            $p^\prime$.insert($u^\prime$)\;$cost+=(|N(u^\prime)|+|N_2^q(u^\prime)|)$\;}
        }
        \uIf{$cost > M$}{$P$.append($p$),\textbf{break}\;
        }
        \lElse{$p$.insert($u$), $U$.delete($u$)}
    }
    
}
\Return{$P$}\;
\end{algorithm}
\setlength{\textfloatsep}{0pt plus 0pt minus 0pt}

\textsf{BCPar} takes as inputs the vertex set $U$ of the selected layer, neighbor structure $N(\cdot)$ and $N_2^q(\cdot)$, the memory budget $M$, and returns the partitioning result $P$. 
The algorithm first computes the weight $w$ and the average weight $avg_w$ for each vertex in $U$ (lines 1--2). 
Then, the algorithm sorts the vertices in $U$ in descending order and stores them in array $L$ (line 3). 
% The variable $P$ is used to store the partitioning result which is initialized to be empty (line 4). 

\textsf{BCPar} leverages a while-loop to partition the vertices (lines 5--22). 
In each iteration, four variables are initialized: $p$ for storing the vertices in the current partition, $p^\prime$ for recording the vertices in $p$ and their 2-hop neighbors, $Q$ for storing the in-neighbors of vertices in $p^\prime$, and $cost$ for tracking partition size (line 6). 
The $cost$ of partition $p$ is the total number of the vertices in $p^\prime$ and the corresponding 1-hop and 2-hop neighbors. 
{Subsequently, \textsf{BCPar} prioritizes the vertex $u$ with the maximal average weight to initialize $p$ (line 7). The rationale behind this is that if the neighbors of the vertex with the maximal average weight can be shared by other vertices, there exists a heightened likelihood of achieving better compression gains.}
\textsf{BCPar} uses a hash table $p^\prime$ to record $u$ and its 2-hop neighbors, subsequently adjusting the cost after the insertion of $u$ into the current partition (lines 8--10). 
{The algorithm iteratively appends vertices to the ongoing partition $p$ (lines 11--23).} 
For each vertex $u^\prime$ in $p^\prime$, the algorithm increments the weight of their corresponding in-neighbors by $w(u^\prime)$, signifying that if vertex $v$ is inserted, the cost is reduced by $w(u^\prime)$ (lines 12--16).
$Q$ is established as a max-heap consisting of the in-neighbors of vertices in $p^\prime$. {The algorithm then selects the vertex with maximal weight in $Q$ as a candidate for the current partition $p$ (line 17). }
The reason behind this is that inserting the vertex with the maximum weight can reduce the partition size to the greatest extent possible.
Subsequently, \textsf{BCPar} checks if inserting $u^\prime$ satisfies the memory requirement $M$. If not, the ongoing partition $p$ is appended to the result set $P$, concluding the current search. Otherwise, $u$ is inserted into $p$, and the exploration continues.
% Subsequently, the algorithm verifies whether inserting $u^\prime$ complies with the memory requirement $M$. 
% If the criterion is not met, the ongoing partition $p$ is appended to the result set $P$, concluding the current search. Otherwise, $u$ is inserted into $p$, and the algorithm continues the exploration.

\textsf{BCPar} holds two primary advantages. 
Firstly, the sharing of multiple neighbors among vertices within the same partition significantly reduces the partition size. 
Secondly, it enhances load balancing, leading to improved performance.

% Given the device memory budget $M$ for holding the input graph $G$, we aim to partition $G$ into $k$ small subgraphs $P = \{ g_1,g_2,...,g_k \}$, such that $g_1 \cap g_2 \cap ... \cap g_k = G$ and $g_k \leq M$. In order to reduce the cost of data transfer, we resort to minimizing the following object function:

% \begin{equation}
%     cost(P) = \sum_k cost(g_k), \quad s.t. \;cost(g_k) \leq M \; and\; g_k \in P
% \end{equation}

% \input{Method} 
% \vspace{-2mm}
\section{Evaluation}
\label{sec:evaluation}
% \vspace{-2mm}

\begin{table}[tbp]
    \vspace{1mm}
    \begin{center}   
        \caption{{Detailed descriptions of datasets.}}  
        \setlength{\tabcolsep}{2.4pt}
        % \vspace{-2mm}
        \label{datasetdetail} 
        % \begin{tabular}{|p{2.3cm}<{\centering}||p{0.9cm}<{\centering}|p{0.9cm}<{\centering}|p{1.2cm}<{\centering}|p{0.45cm}<{\centering}|p{0.45cm}<{\centering}|}   
        \begin{tabular}{|c||c|c|c|c|c|}
        \hline   \textbf{Datasets} & $|$\textbf{\textit{U}}$|$ & $|$\textbf{\textit{V}}$|$ &  $|$\textbf{\textit{E}}$|$ & {$\bar{d}_U$} & {$\bar{d}_V$}     \\ \hline
        \hline   Youtube (\textit{YT}) & 94,238 & 30,087 & 293,360 & {3.11} & {9.75}\\ 
        \hline   Bookcrossing (\textit{BC}) & 77,802 & 185,955 & 433,652 & {5.57} & {2.33} \\
        \hline   Github (\textit{GH}) & 56,519 & 120,867 & 440,237 & {7.79} & {3.64}\\
        \hline   StackOverflow (\textit{SO}) & 545,196 & 96,680 & 1,301,942 & {2.39} & {13.47} \\
        \hline   Yelp (\textit{YL}) & 31,668 & 38,048 & 1,561,406 & {49.31} & {41.04}\\
        \hline   IMDB (\textit{ID}) & 303,617 & 896,302 & 3,782,463 & {12.46} & {4.22}\\
        \hline   Lastfm (\textit{LF}) & 359,349 & 160,168 & 17,559,162 & {48.86} & {109.63}\\
        \hline   Edit\_fr (\textit{FR}) & 16,874 & 3,416,271 & 23,443,737 & {1,389.34} & {6.86}\\
        \hline   Orkut (\textit{OR}) & 2,783,196 & 8,730,857 & 327,037,487 & {117.50} & {37.45}\\
        \hline   Synthetic 1 (\textit{S1}) & 6,720 & 5,300 & 207,146 & {30.83} & {39.08} \\
        \hline   Synthetic 2 (\textit{S2}) & 12,720 & 11,100 & 220,651 & {17.35} & {19.88} \\
        \hline   
        \end{tabular}   
    \end{center}
    % \vspace{-1mm}
\end{table}
 
In this section, we evaluate the performance of \textsf{GBC} and conduct a comparative evaluation with the state-of-the-art $(p,q)$-biclique counting algorithms. 

% \vspace{-2.5mm}
\subsection{Experimental Setup}
% \vspace{-1mm}

\textbf{Evaluated Methods.} We compare \textsf{GBC} with three algorithms: (1) \textsf{GBL}: the GPU baseline designed in~\S~\ref{gpubaseline}; (2) \textsf{BCL}: the state-of-the-art CPU solution~\cite{yang2021p}; (3) \textsf{BCLP}: the parallelized \textsf{BCL} on CPU~\cite{yang2021p}. 

\textbf{Platform.} All experiments are conducted on an Ubuntu 20.04.2 server, featuring an Intel Core i9-10900K 3.70GHz CPU and an Nvidia GeForce RTX 3090 GPU. The GPU has 82 SMs and $10,496$ cores. We implement \textsf{GBC} in C++ under Nvidia CUDA 11.2. \textsf{BCLP} is executed with 16 threads.

\textbf{Datasets.} Table~\ref{datasetdetail} summarizes the 11 datasets used for experiments. 
The first 9 datasets\footnote{http://konect.cc/networks} are real datasets with 8 of them extensively employed in related research~\cite{sanei2018butterfly, yang2021p}. The large dataset Orkut (\textit{OR}) is prepared for evaluating scalability where data cannot entirely fit into the global memory. 
Besides, we generate 2 synthetic datasets ($S1$ and $S2$) with more 2-hop neighbors than those in the real datasets, resulting in increased computational load and uneven workloads.
{The synthetic datasets are generated as follows:
(1) fix the size of $U$ and $V$, 
(2) determine the number of 2-hop neighbors for layer $U$ according to the power-law distribution, then artificially adjust it to be slightly larger than that of the real datasets used.
(3) randomly select neighbors from $V$ for $U$ based on the generated number of 2-hop neighbors.}
% The detailed descriptions of the datasets are summarized in Table~\ref{datasetdetail}. 

\textbf{Queries.} We set $(p+q)$ in the range of 8 to 24, with 16 serving as the default value.

\begin{figure*}[tb]
    \centering
    \includegraphics[width=0.35\textwidth]{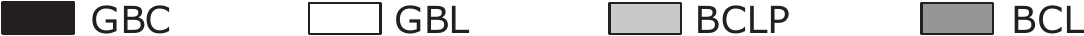}\\
    % \vspace{-2mm}
    
    \hspace{-2mm}
    \subfigure{
    \includegraphics[width=0.191\textwidth]{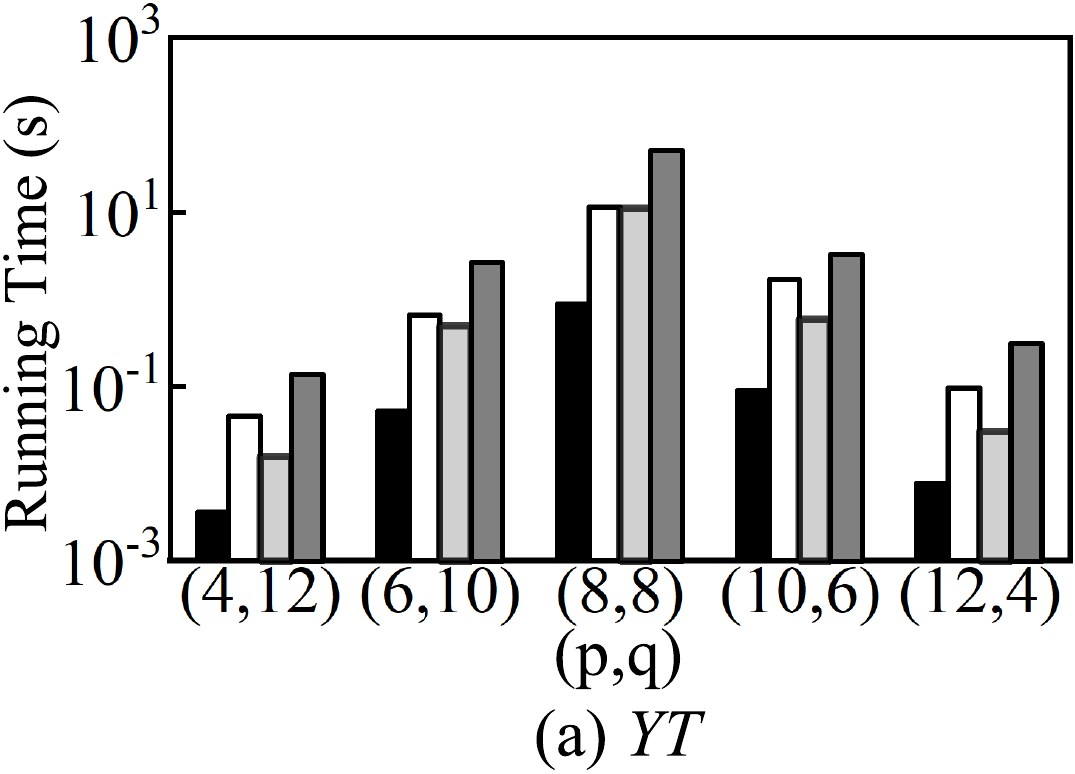}}  
    \hspace{-2.3mm}
    \subfigure{
    \includegraphics[width=0.191\textwidth]{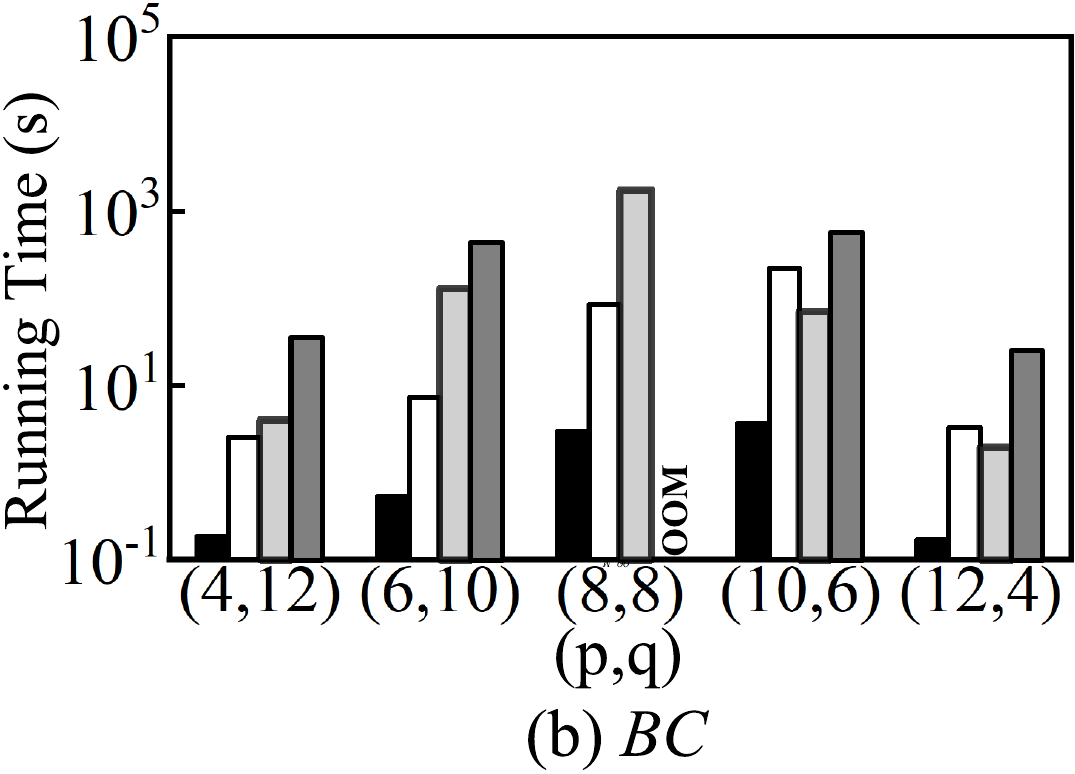}}  
    \hspace{-2.3mm}
    \subfigure{
    \includegraphics[width=0.191\textwidth]{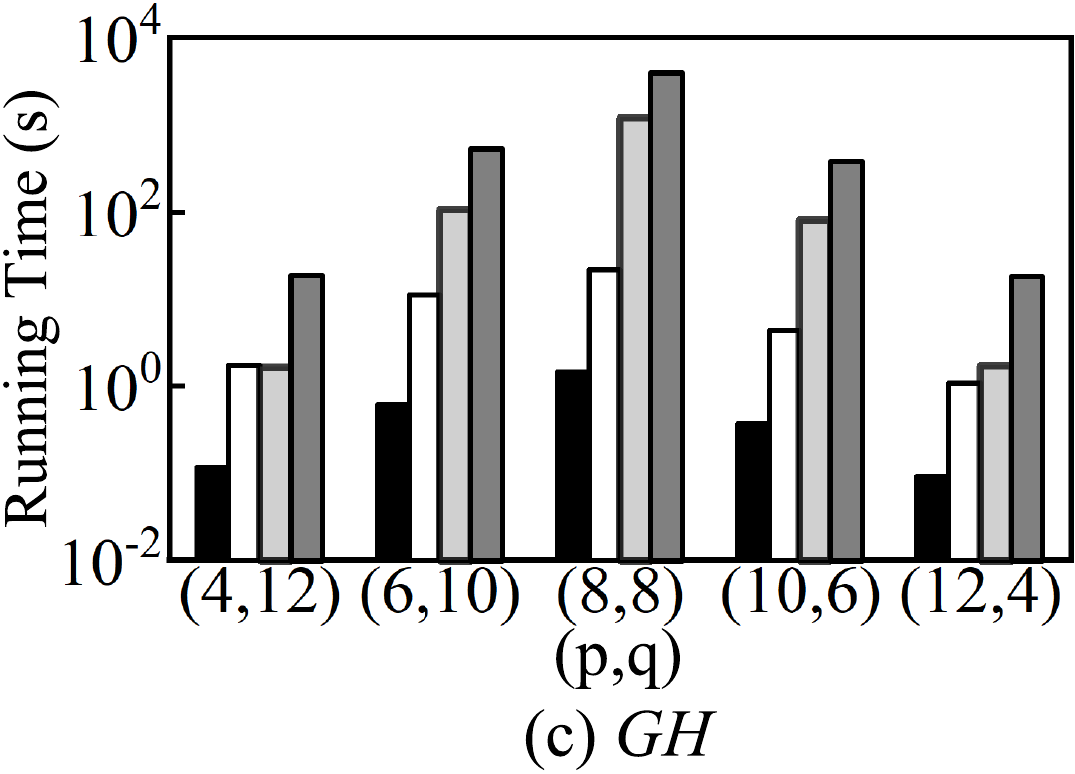}}  
    \hspace{-2.3mm}
    \subfigure{
    \includegraphics[width=0.191\textwidth]{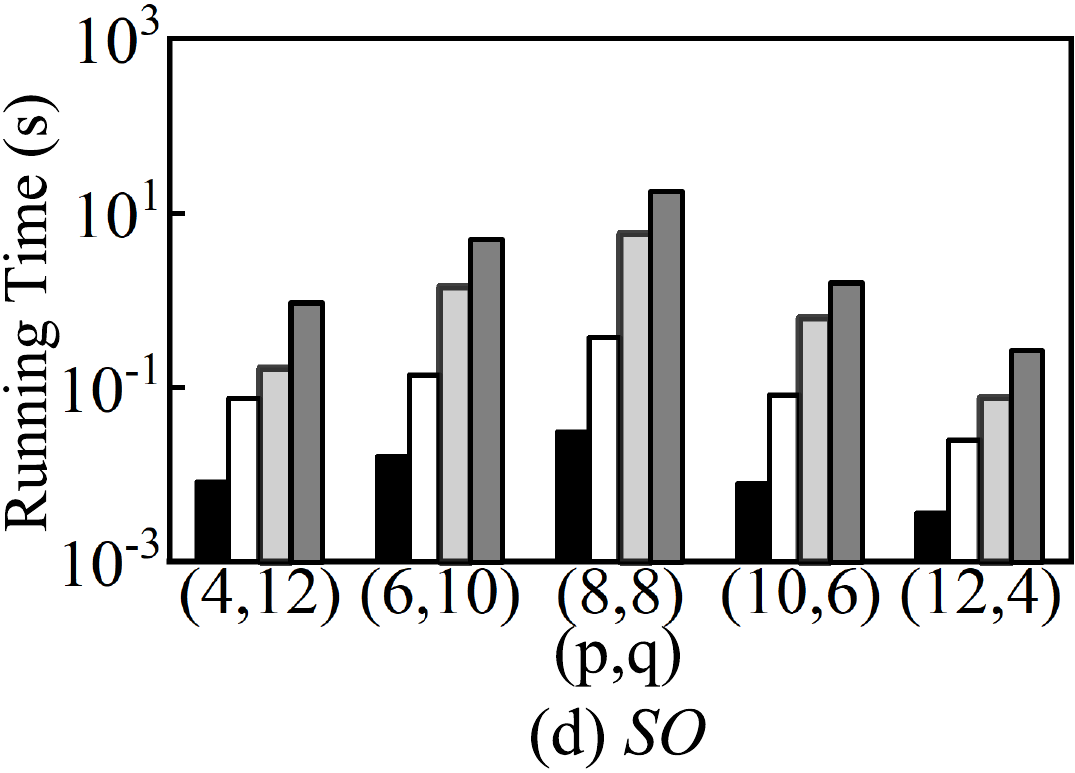}}
    \hspace{-2.3mm}
    \subfigure{
    \includegraphics[width=0.191\textwidth]{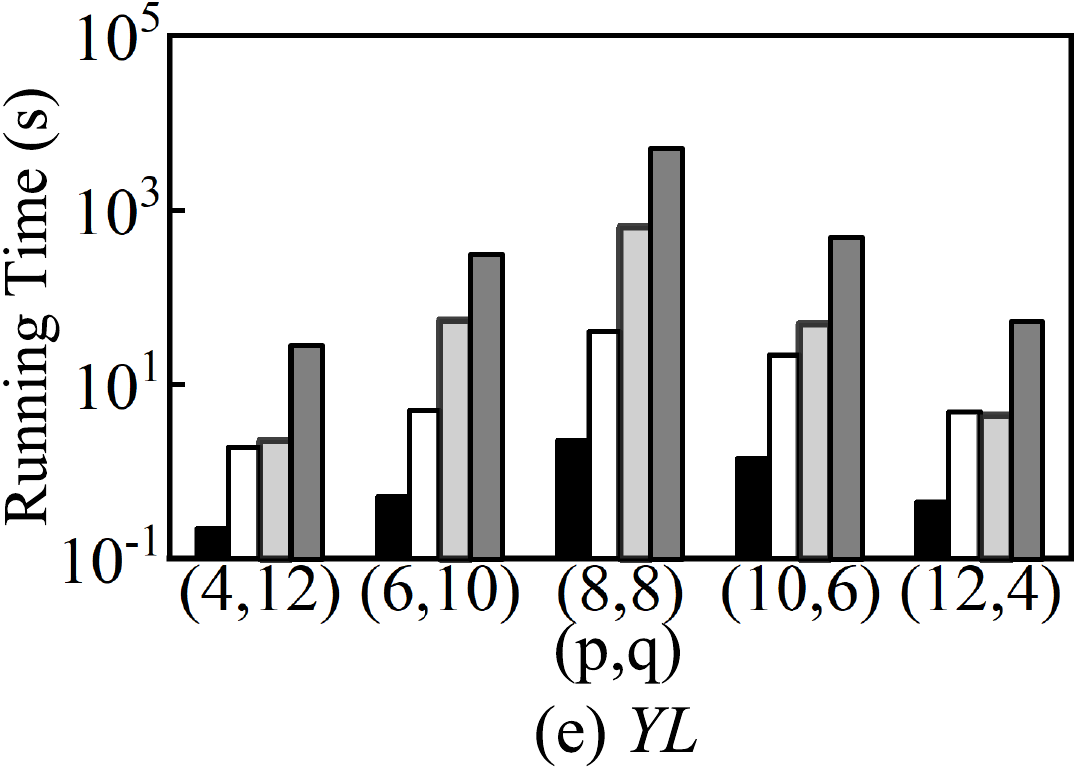}}
    \vspace{-2mm}
    
    \hspace{-3mm}
    \subfigure{
    \includegraphics[width=0.191\textwidth]{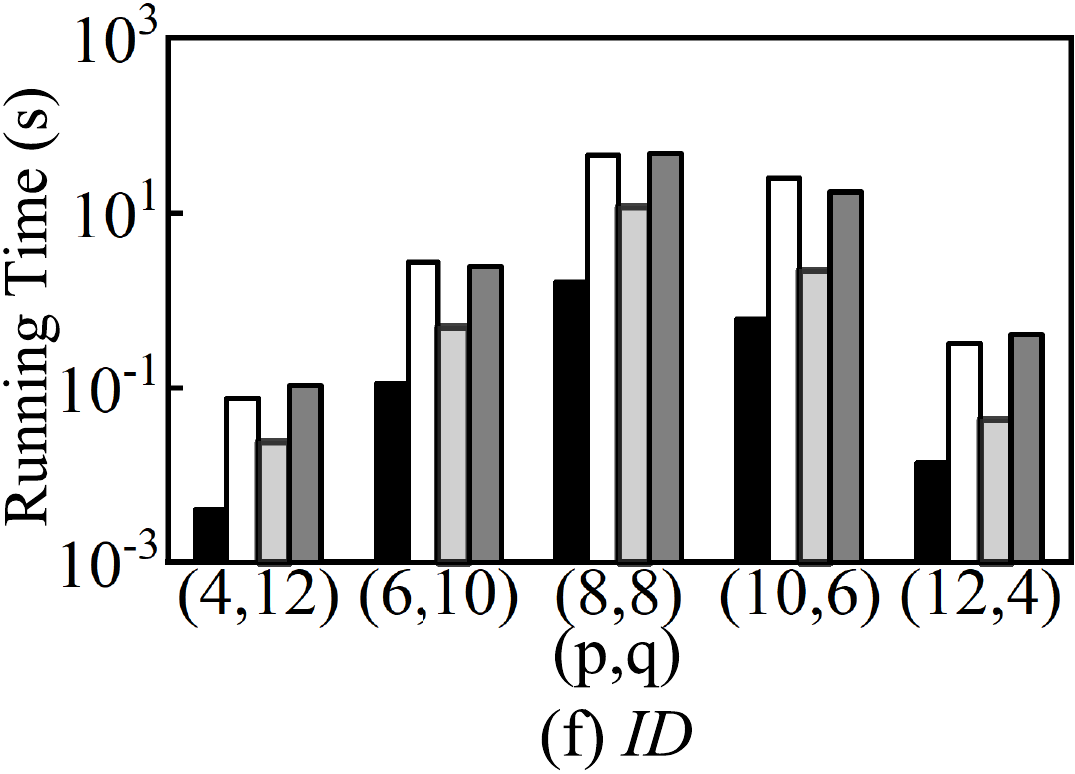}}  
    \hspace{-2.3mm}
    \subfigure{
    \includegraphics[width=0.191\textwidth]{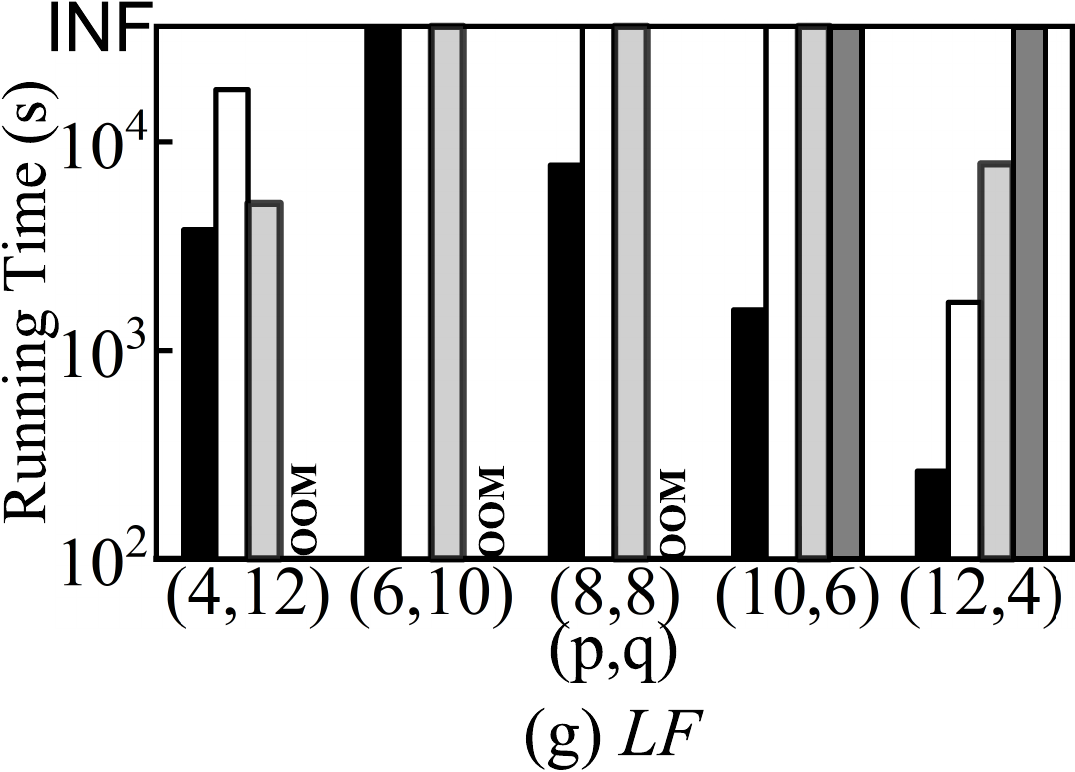}}  
    \hspace{-2.3mm}
    \subfigure{
    \includegraphics[width=0.191\textwidth]{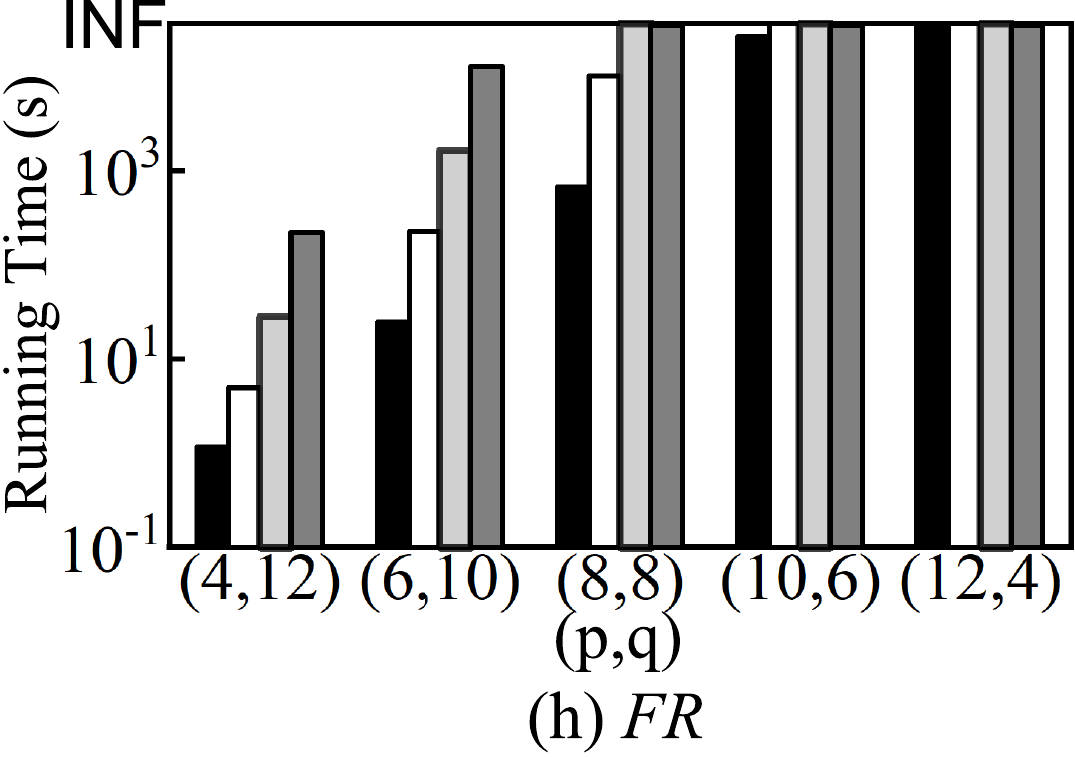}}  
    \hspace{-2.3mm}
    \subfigure{
    \includegraphics[width=0.191\textwidth]{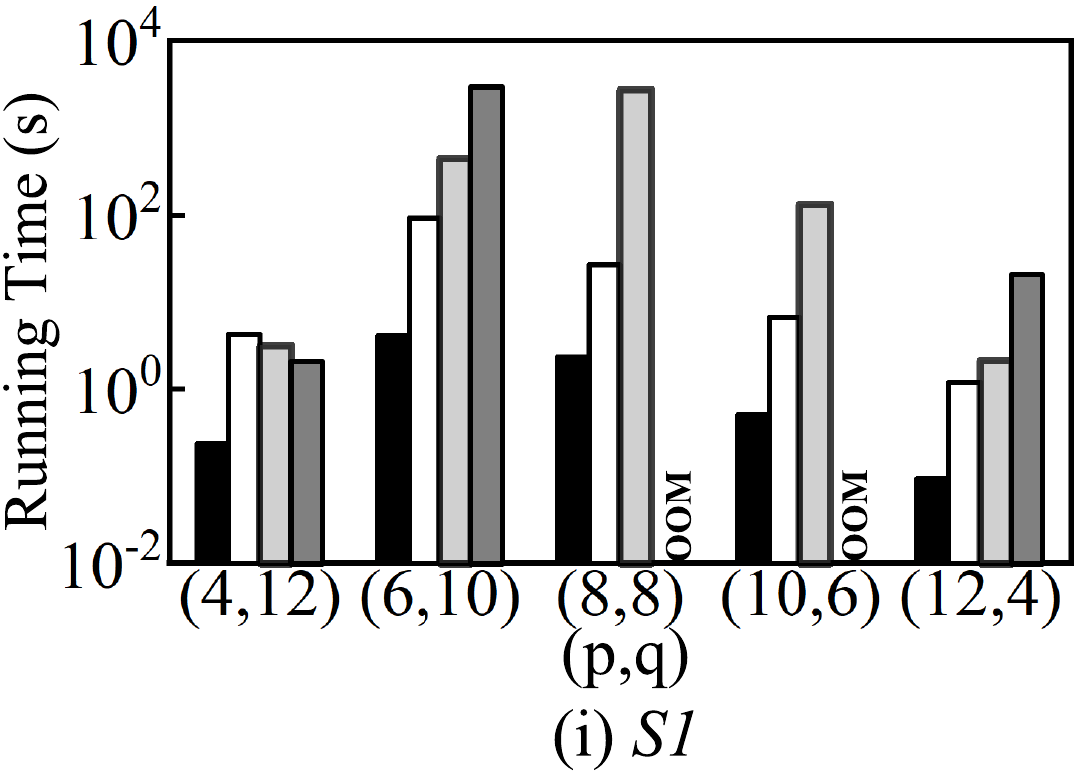}}
    \hspace{-2.3mm}
    \subfigure{
    \includegraphics[width=0.191\textwidth]{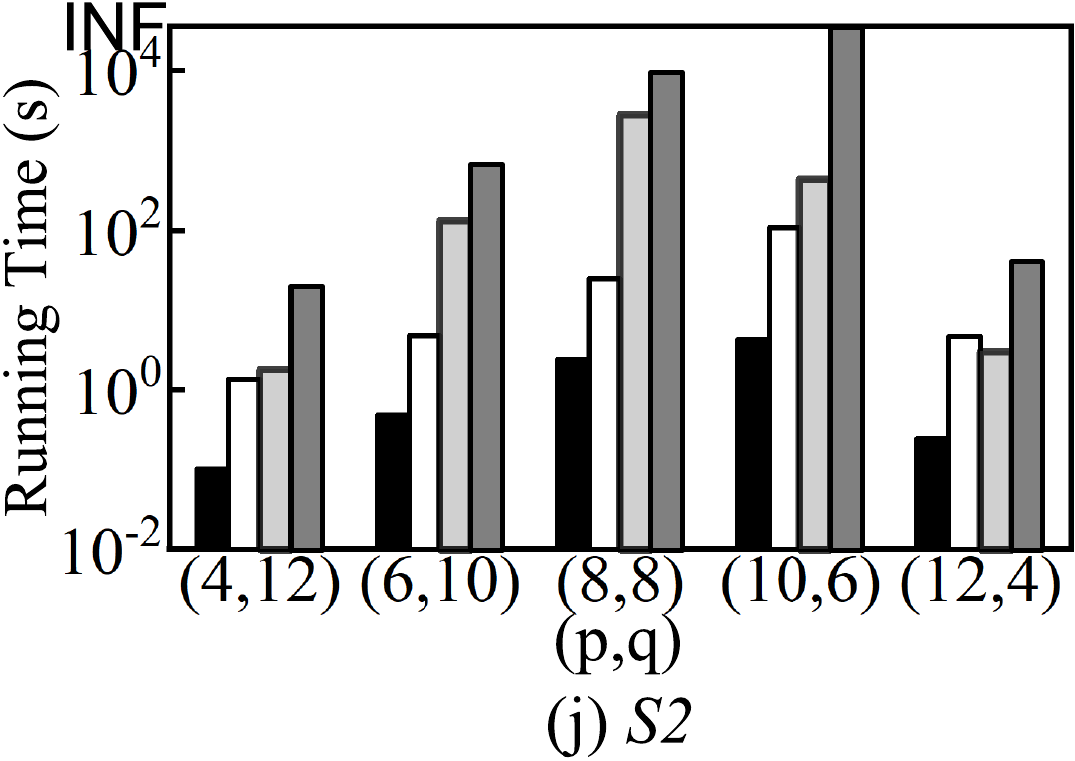}}
    
    % \hspace{-5mm}
    \vspace{-4mm}
    \caption{Overall performance in comparison to baselines (INF: running times beyond 10 hours; OOM: out of memory).}
    \vspace{-2mm}
    \label{overallperform}
\end{figure*}

% \vspace{-3mm}
\subsection{Overall Performance}
\label{sec:overallperformance}
\vspace{-2mm}

Figure~\ref{overallperform} illustrates the performance of different methods. In this context, methods exceeding a runtime of 10 hours are terminated and their values are set to INF.

First, \textsf{GBC} significantly outperforms the compared methods in all cases. 
On average, \textsf{GBC} achieves speed enhancements of 505.3$\times$ over \textsf{BCL}, 146.7$\times$ over \textsf{BCLP} and 15.7$\times$ over \textsf{GBL}. The considerable improvement of \textsf{GBC} stems from the full exploitation of GPU's computing resources.
In the best cases, \textsf{GBC} achieves a remarkable 
% 2782.8$\times$ speedup over \textsf{BCL} ($p = q = 8$ on \textit{GH}), an 
836.7$\times$ acceleration over \textsf{BCLP} ($p = q = 8$ on \textit{GH}) 
% and a 59.6$\times$ speedup over \textsf{GBL} ($p = 10, q = 6$ on \textit{BC}) 
on real datasets. Meanwhile, on synthetic datasets, \textsf{GBC} realizes an extraordinary 
% 4163.8$\times$ speedup over \textsf{BCL} ($p=q=8$ on \textit{S2}), a 
1217.6$\times$ acceleration over \textsf{BCLP} ($p=q=8$ on \textit{S2}).
% and a 25.8$\times$ speedup over \textsf{GBL} ($p=6,q=10$ on \textit{S2}) on synthetic datasets. 
Compared to real datasets, \textsf{GBC} exhibits superior performance gain on synthetic datasets, owing to their elevated computational overheads. By harnessing the considerable parallelism of GPUs, \textsf{GBC} adeptly tackles demanding computations with notable efficacy.

Second, \textsf{GBL}'s performance sometimes lags behind that of \textsf{BCLP} (e.g., on \textit{TY} and \textit{ID}). This discrepancy arises from \textsf{GBL} being a na\"ive migration of the basic model to the GPU platform, without advanced optimizations tailored for the GPU architecture. In contrast, through the incorporation of multiple innovative designs customized for GPUs, \textsf{GBC} attains optimal performances across all scenarios.
The notable progress demonstrated by \textsf{GBC} underscores the importance of carefully adapting algorithms to harness the unique capabilities and parallel processing power of GPUs.

Third, varying parameters $p$ and $q$ yield disparate execution times. 
Generally, when either $p$ or $q$ is small, the depth of the search tree remains correspondingly shallow, leading to lower computation costs.
Another noteworthy observation is that the speedup ratio of \textsf{GBC} tends to decrease as the disparity between the parameters increases.
{This is because the proposed techniques exhibit greater effectiveness with higher search trees, implying more computational workloads.}
Deeper levels in the high search tree result in smaller sizes of $C_R$ and $C_L$, along with fewer rounds of intersection operations, thus making the optimizations considerably more impactful. 
Furthermore, workload imbalance is less pronounced when the height of the search trees is small.
% \textcolor{red}{And large $p$ or $q$ results in longer adjacency lists which contributes to denser representation for hierarchical truncated bitmap. }
% As a consequence, the speedup ratio of $\mathsf{GBC}$ will decrease when the parameter decreases, as the algorithm's optimization mechanisms are less effective in scenarios with shallower search trees and lighter workloads. 

Fourth, the improvement of \textsf{GBC} appears conservative on some datasets. 
For instance, \textsf{GBC} achieves a maximum speedup of only 7.4$\times$ faster than \textsf{BCLP} on \textit{ID}. {As observed in Figure~\ref{overallperform}, the running time of \textsf{BCLP} is relatively brief among those datasets with lower speedup ratios of \textsf{GBC}. 
This suggests that the CPU suffices to count $(p,q)$-bicliques in these datasets within a relatively short duration, thereby diminishing the necessity to offload the computation to the GPU. 
In contrast, datasets necessitating prolonged processing time on the CPU are more suited for GPU acceleration. The high parallelism offered by the GPU proves particularly advantageous for datasets with substantial computational demands, substantially reducing processing time.}

% \vspace{-2mm}
\subsection{Scalability Evaluation}
% \vspace{-2mm}

We proceed to evaluate the scalability concerning different methods by varying the biclique size, i.e., the value of $(p+q)$, from 8 to 24 with an increment of 4, where $p=q=(p+q)/2$. Figure~\ref{scalabilityfigure} shows the results on five representative datasets.

The first observation is that \textsf{GBC} consistently outperforms the compared methods across all test cases, exhibiting a substantial improvement ranging from 2.4$\times$ to 6298.1$\times$. The performance improvement of \textsf{GBC} becomes more notable with an increase in computational workload, as exemplified in datasets \textit{BC} and \textit{GH}.
The performance enhancement stems from \textsf{GBC}'s capability to fully unleash the computing resources of the GPU, which is specifically suitable for intensive computing.
The second observation is that the running times of CPU solutions increase initially and then decrease with the growth of biclique size, whereas GPU methods generally exhibit comparable changes or continuous decreases as the queries grow in size.
This is because \textsf{GBC} effectively balances the workload of parallel threads, preventing excessively long-running threads from becoming bottlenecks.

% varying the biclique size results in disparate running times. 
% The second observation is that \textsf{GBC} consistently outperforms the compared methods across all test cases, exhibiting a substantial improvement ranging from 2.4$\times$ to 6298.1$\times$. The performance improvement of \textsf{GBC} becomes more notable with an increase in computational workload, as exemplified in datasets \textit{BC} and \textit{GH}.
% Particularly, \textsf{GBC} achieves an average speedup of 405.2$\times$ over \textsf{BCLP} and an exceptional 6298.1$\times$ speedup on \textit{BC} when $p+q=20$. 
% Detailed observation of Figure~\ref{scalabilityfigure} reveals that the running time of \textsf{BCLP} in \textit{BC} and \textit{GH} is notably longer when $p+q$ equals 16 and 20. In these instances, the algorithms need to search a larger space, resulting in higher time complexity. In contrast, \textsf{GBC} effectively harnesses the advantages of GPU in intensive computing, yielding superior performance. 

% \textcolor{red}{[LS: more details?]}

\begin{figure*}[tb]
    \centering
    \includegraphics[width=0.35\textwidth]{figures/legend.pdf}\\
    % \vspace{-2mm}
    
    \hspace{-2mm}
    \subfigure{
    \includegraphics[width=0.191\textwidth]{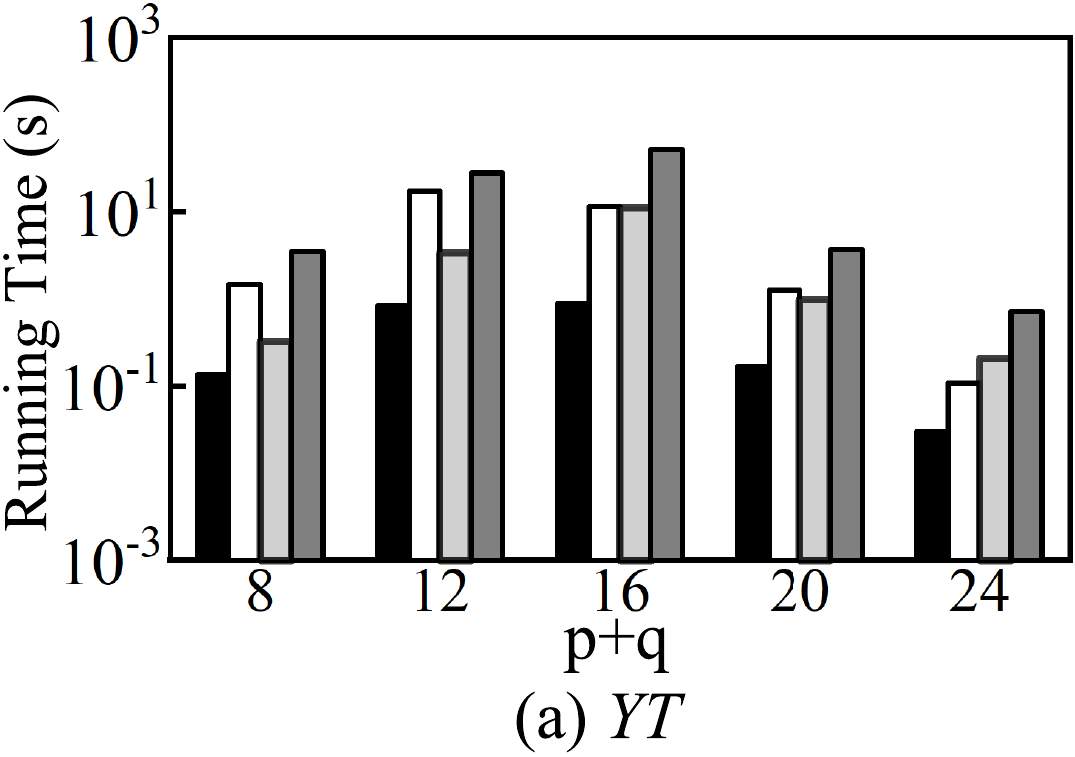}}  
    \hspace{-2.3mm}
    \subfigure{
    \includegraphics[width=0.191\textwidth]{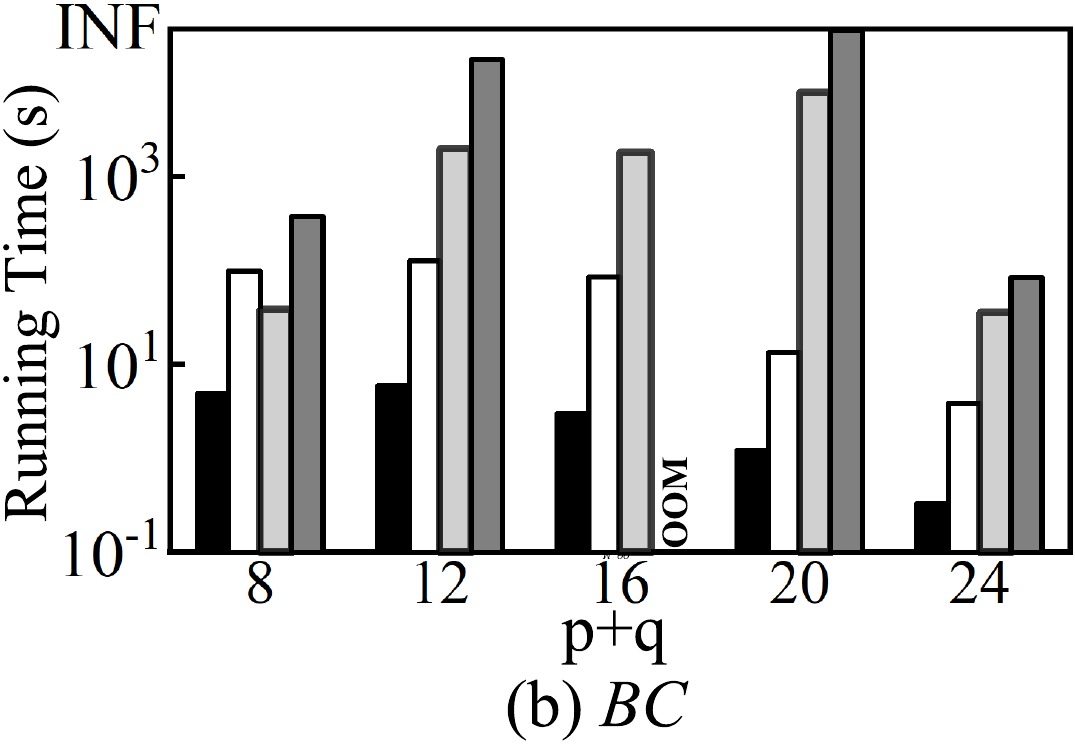}}  
    \hspace{-2.3mm}
    \subfigure{
    \includegraphics[width=0.191\textwidth]{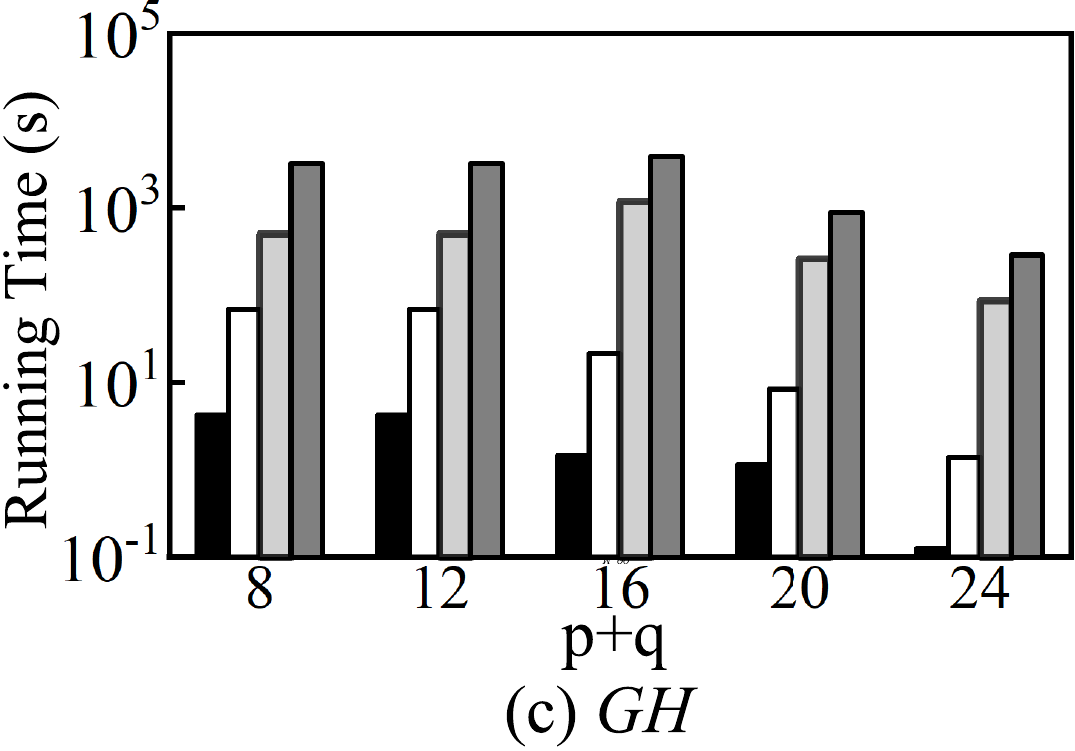}}  
    \hspace{-2.3mm}
    \subfigure{
    \includegraphics[width=0.191\textwidth]{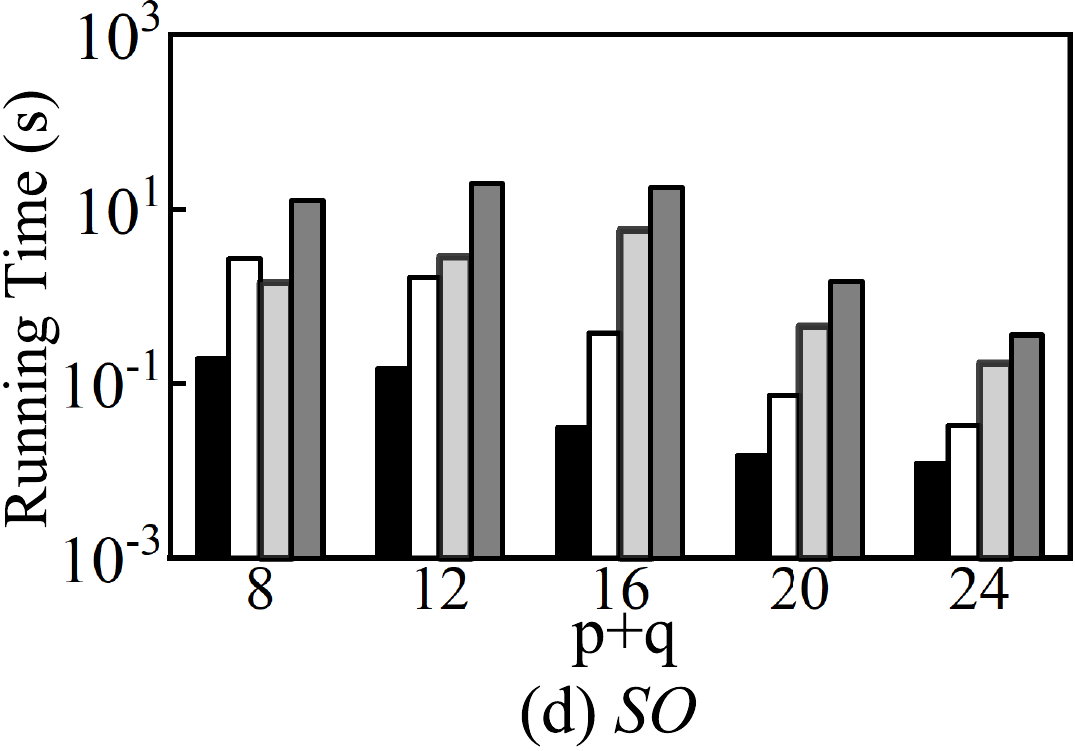}}
    \hspace{-2.3mm}
    \subfigure{
    \includegraphics[width=0.191\textwidth]{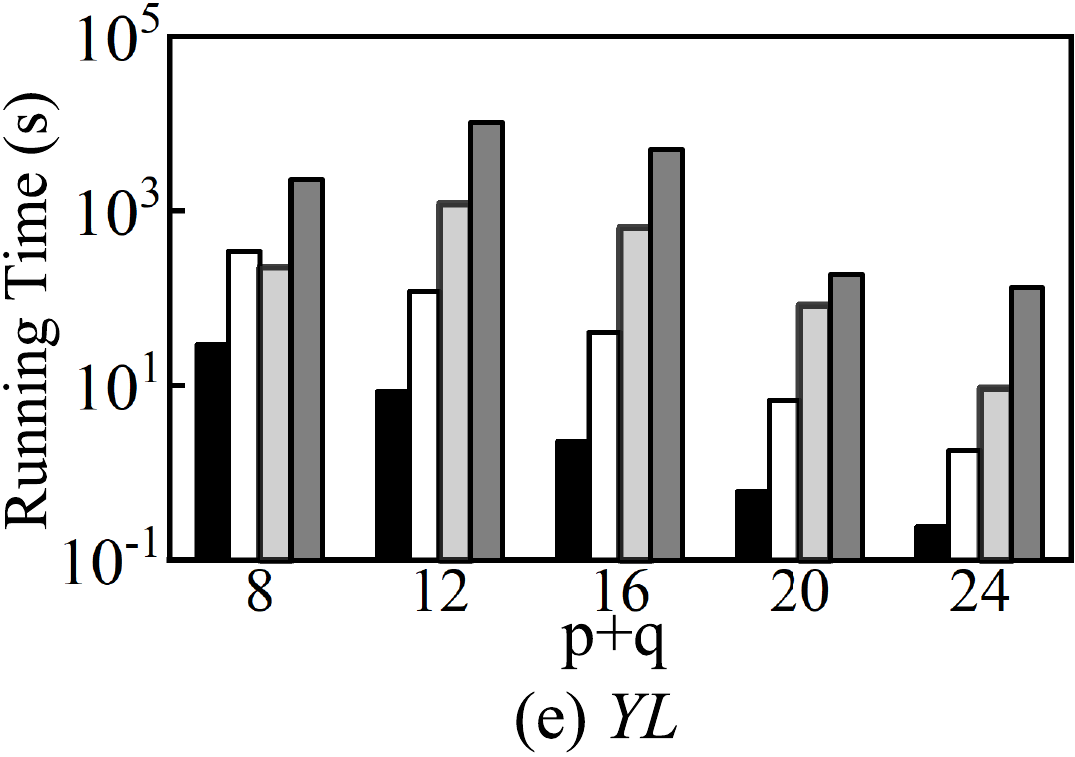}}
    % \vspace{-3mm}
    
    % \hspace{-2mm}
    % \subfigure[\textit{ID}]{
    % \includegraphics[width=0.208\textwidth]{figures/imdbablationh.pdf}}  
    % \hspace{-2.3mm}
    % \subfigure[\textit{LF}]{
    % \includegraphics[width=0.191\textwidth]{figures/lastfmablationh.pdf}}  
    % \hspace{-2.3mm}
    % \subfigure[\textit{FR}]{
    % \includegraphics[width=0.191\textwidth]{figures/editfrablationh.pdf}}  
    % \hspace{-2.3mm}
    % \subfigure[\textit{S1}]{
    % \includegraphics[width=0.191\textwidth]{figures/githubnew2ablationh.pdf}}
    % \hspace{-2.3mm}
    % \subfigure[\textit{S2}]{
    % \includegraphics[width=0.191\textwidth]{figures/generate6ablationh.pdf}}
    
    % \hspace{-5mm}
    \vspace{-4mm}
    \caption{Scalability evaluation vs. query size $(p+q)$ ($p=q=(p+q)/2$).}
    \vspace{-2mm}
    \label{scalabilityfigure}
\end{figure*}

\begin{figure*}[tb]
    \centering
    % \includegraphics[width=0.35\textwidth]{figures/le.pdf}\\
    % \vspace{-2mm}
    
    \hspace{-2mm}
    \subfigure{
    \includegraphics[width=0.191\textwidth]{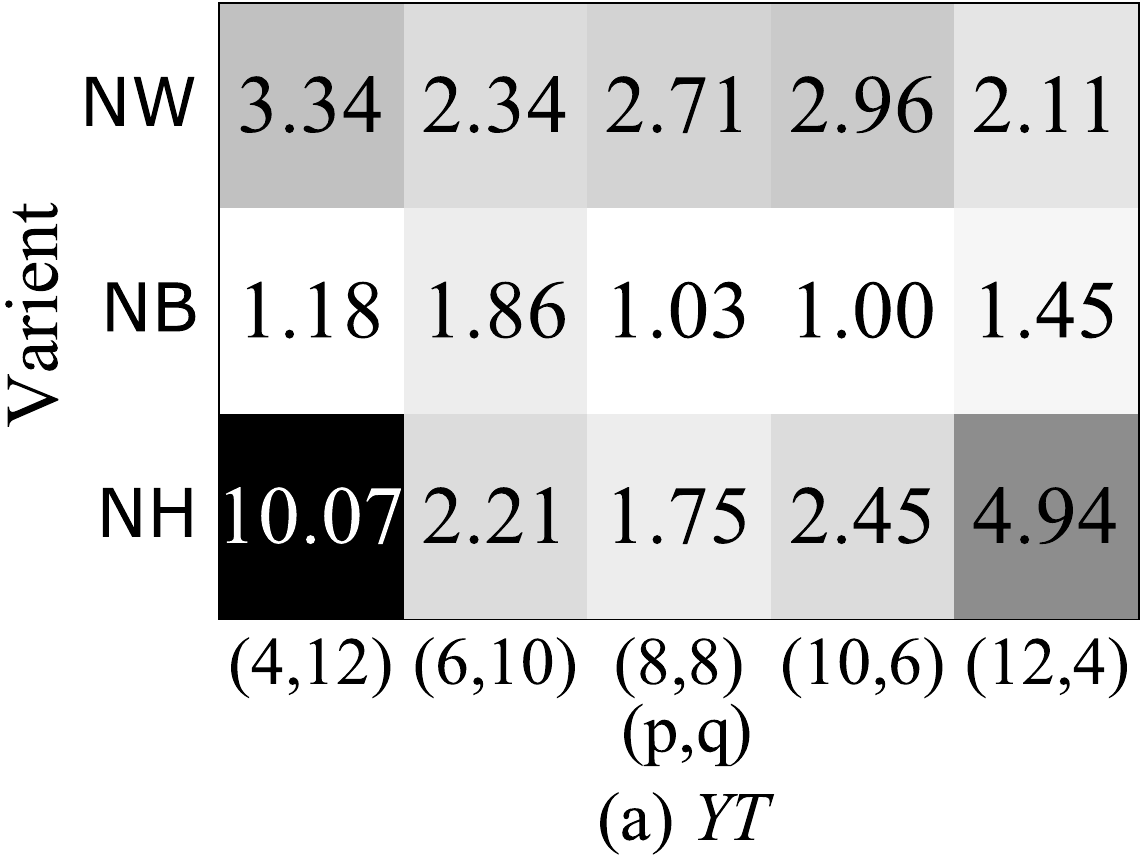}}  
    \hspace{-2.3mm}
    \subfigure{
    \includegraphics[width=0.191\textwidth]{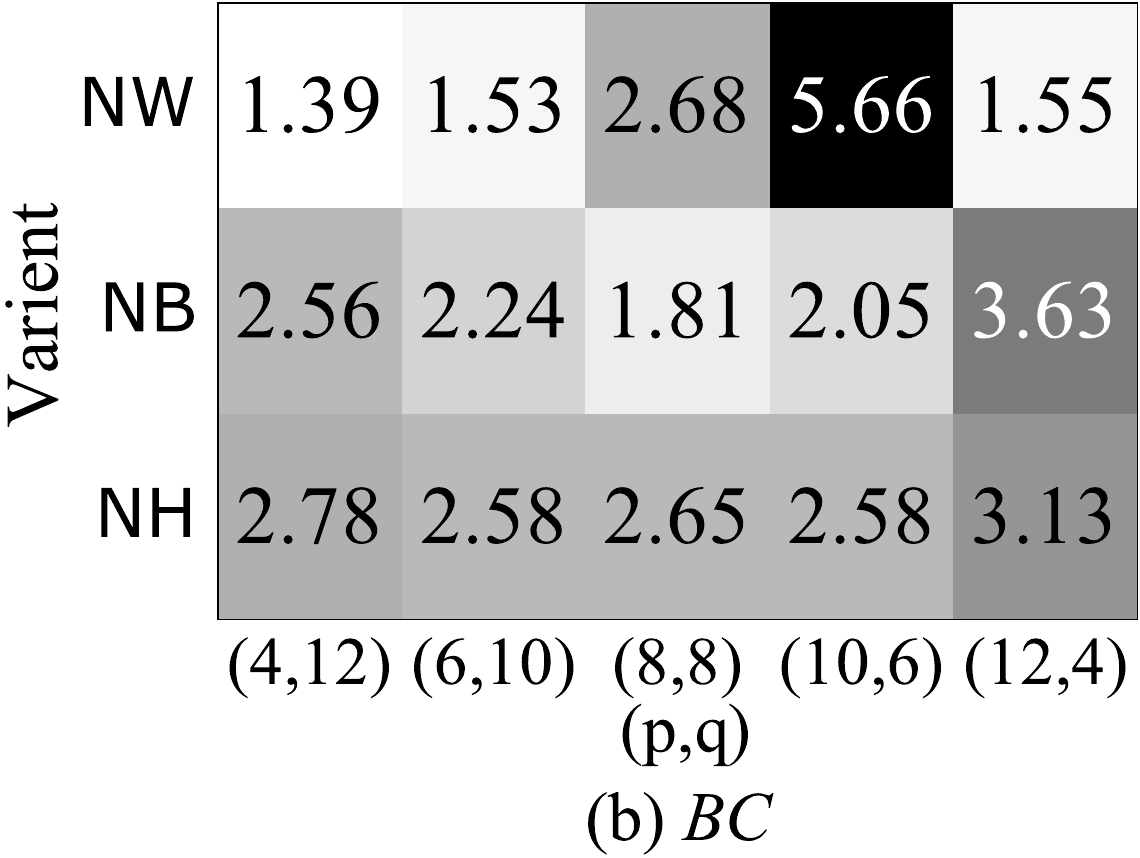}}  
    \hspace{-2.3mm}
    \subfigure{
    \includegraphics[width=0.191\textwidth]{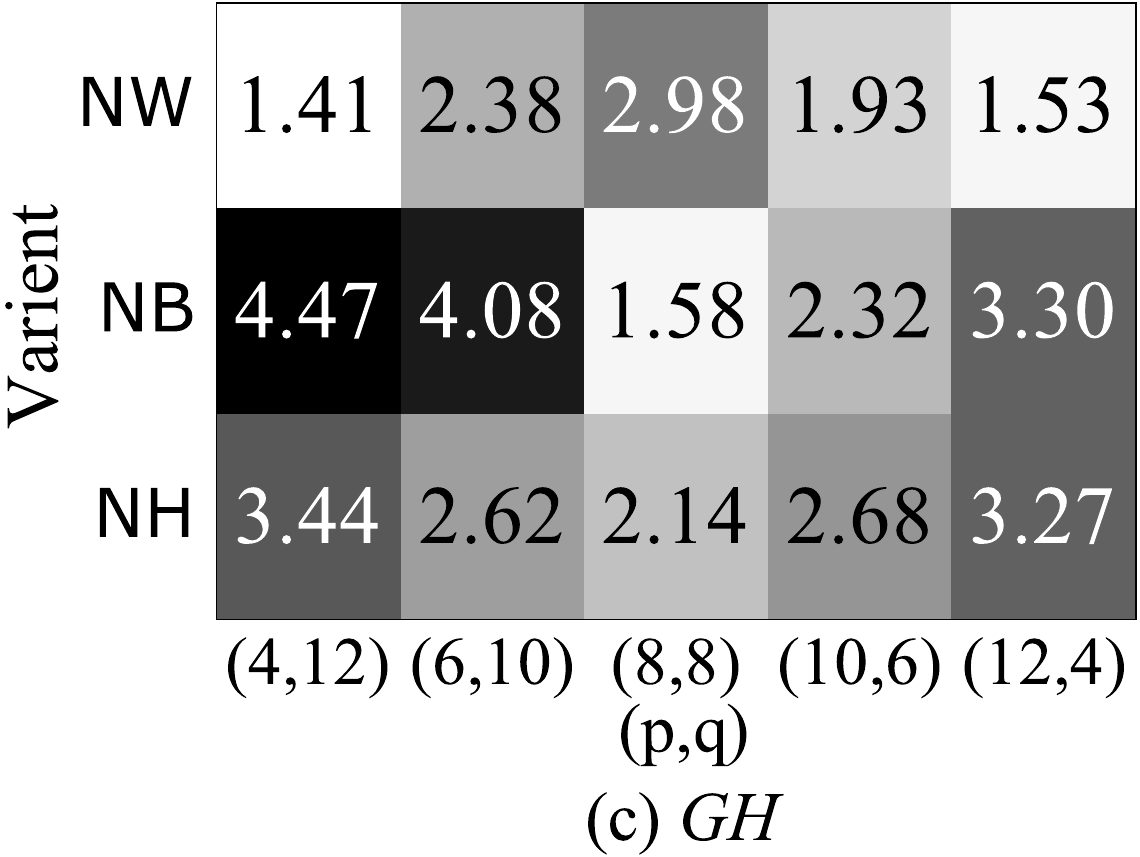}}  
    \hspace{-2.3mm}
    \subfigure{
    \includegraphics[width=0.191\textwidth]{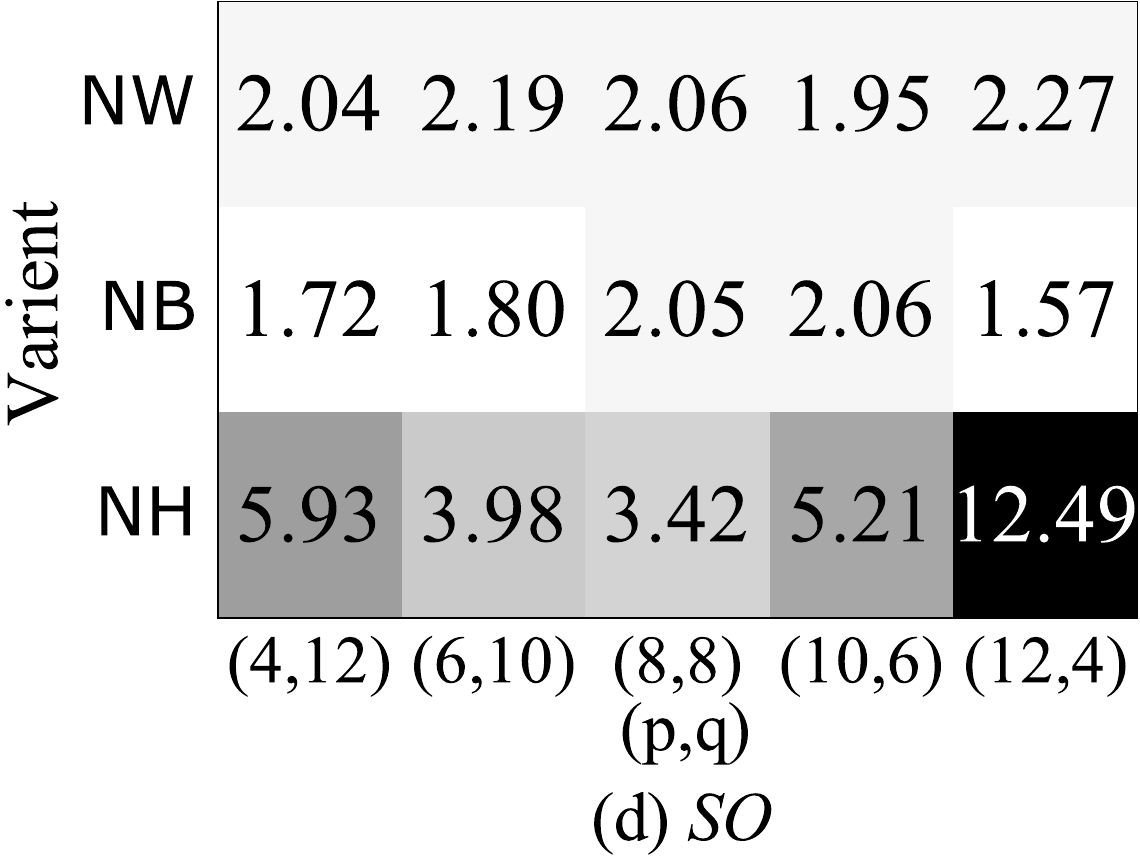}}
    \hspace{-2.3mm}
    \subfigure{
    \includegraphics[width=0.191\textwidth]{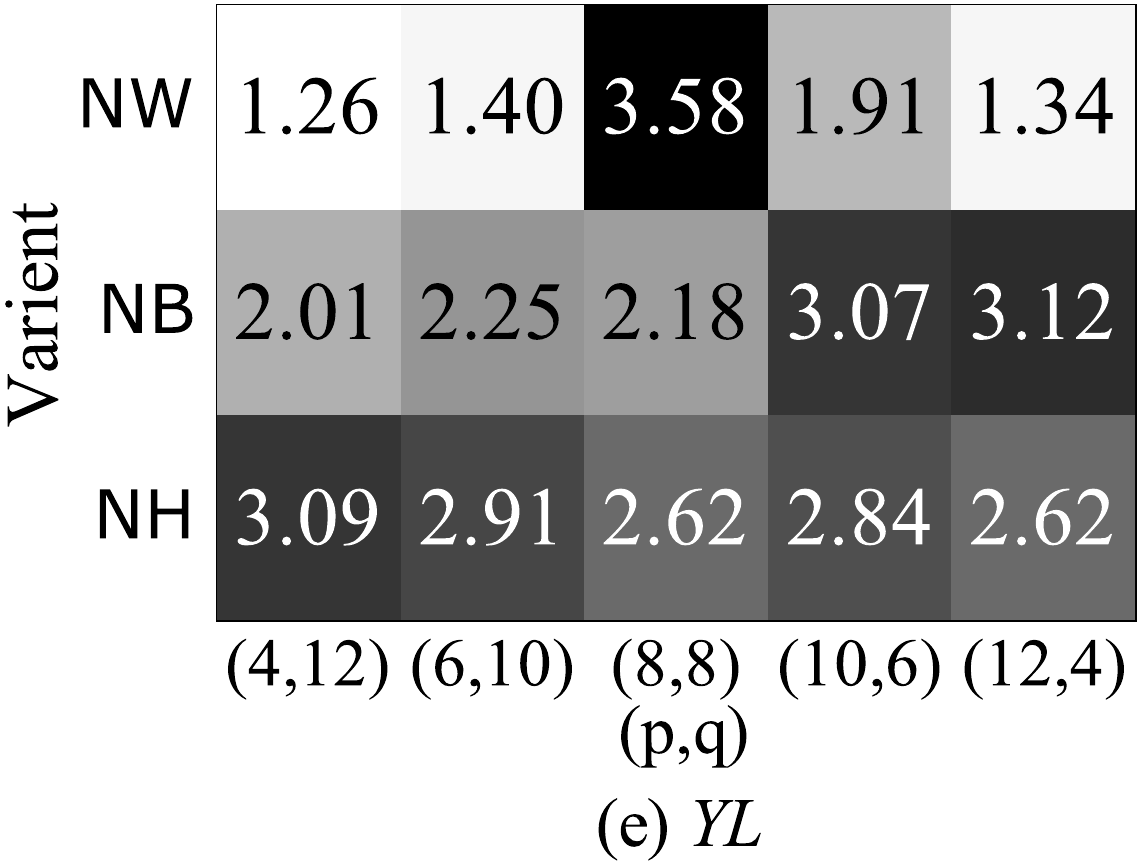}}
    % \vspace{-3mm}
    
    % \hspace{-2mm}
    % \subfigure[\textit{ID}]{
    % \includegraphics[width=0.208\textwidth]{figures/imdbablationh.pdf}}  
    % \hspace{-2.3mm}
    % \subfigure[\textit{LF}]{
    % \includegraphics[width=0.191\textwidth]{figures/lastfmablationh.pdf}}  
    % \hspace{-2.3mm}
    % \subfigure[\textit{FR}]{
    % \includegraphics[width=0.191\textwidth]{figures/editfrablationh.pdf}}  
    % \hspace{-2.3mm}
    % \subfigure[\textit{S1}]{
    % \includegraphics[width=0.191\textwidth]{figures/githubnew2ablationh.pdf}}
    % \hspace{-2.3mm}
    % \subfigure[\textit{S2}]{
    % \includegraphics[width=0.191\textwidth]{figures/generate6ablationh.pdf}}
    
    % \hspace{-5mm}
    \vspace{-4mm}
    \caption{{Average speedup of \textsf{GBC} compared to its variants (\textsf{NH}: w/o hybrid exploration; \textsf{NB}: w/o HTB \& \textsf{Border}); \textsf{NW}: w/o workload balancing).}}
    % \vspace{-5mm}
    \label{ablationfigure}
    \vspace{-6mm}
\end{figure*}
% \vspace{-0.8cm}

% \vspace{-2mm}
\subsection{{Effect of Individual Optimization}}
\label{ablationstudy}
% \vspace{-2mm}

{
In this section, we evaluate the effect of the proposed optimizations by disabling specific modules of \textsf{GBC}, yielding three comparative variants: (1) \textsf{NH} omits the hybrid DFS-BFS exploration, (2) \textsf{NB} discards HTB and \textsf{Border}, and (3)\textsf{NW} excludes the workload balancing.
Due to space constraints and similar experimental outcomes, we only present the results on five datasets (\textit{YT}, \textit{BC}, \textit{GH}, \textit{SO}, \textit{YL}), as depicted in Figure~\ref{ablationfigure}.}

{\textbf{Effect of Hybrid DFS-BFS Exploration.}
The hybrid DFS-BFS exploration plays a pivotal role in optimizing \textsf{GBC}, whose absence leads to a 3.7$\times$ increase in average runtime. 
The integration of hybrid exploration allows us to fully harness the potential of GPU threads by concurrently applying intersections for multiple vertices during the BFS phase, resulting in more efficient thread utilization and thereby reducing the time required for intersections. This notably enhances \textsf{GBC}'s performance on datasets with a low average degree or when reaching the bottom of the search trees with fewer candidates. 
For example, hybrid exploration leads to 4.2$\times$ and 6.2$\times$ reduction in runtimes for \textsf{GBC} on \textit{YT} and \textit{SO} (with average degrees less than 4), respectively.
We further compare the performance between DFS and DFS-BFS ({Appendix-B in~\cite{fullversion2024}}). On average, DFS-BFS incurs 1.3$\times$ more memory overhead, remaining well below the GPU memory capacity. However, DFS-BFS demonstrates superior performance, being on average 2.2$\times$ faster than DFS, attributed to the effective utilization of parallel threads offered by GPU.}

% \vspace{-0.25mm}
{\textbf{Effect of HTB and \textsf{Border}.}
HTB and \textsf{Border} consistently enhance performance across all datasets, resulting in an average speedup of 2.2$\times$ over using CSR. 
This enhancement lies in that HTB and \textsf{Border} effectively compress adjacency lists, which reduces both element comparisons and memory accesses during binary search. This effect is particularly pronounced for datasets characterized by long adjacency lists, as a large proportion of elements can be compressed. For instance, \textsf{GBC} with HTB and \textsf{Border} achieves a speedup of more than 4$\times$ on \textit{GH} when $p=4(6)$ and $q=12(10)$.}

% \vspace{-0.5mm}
{We further explore \textsf{Border}'s efficiency and compare it with a leading reordering method for unipartite graphs, i.e., \textsf{Gorder}~\cite{WeiYLL16}. We apply \textsf{Border} and \textsf{Gorder} to the graphs respectively and execute \textsf{GBC} with HTB on the reordered ones. Table~\ref{vertexreordertable} reports the experimental results.
First, compared to no reordering, \textsf{Gorder} achieves an average speedup of 2.4$\times$, whereas \textsf{Border} represents a notably higher average speedup of 3.1$\times$, affirming the effectiveness of vertex reordering. 
Second, \textsf{Border} exhibits superior performance gain over \textsf{Gorder} across all datasets, averaging 37.0\% faster than \textsf{Gorder}. This improvement is attributed to \textsf{Border}'s specialization in optimizing the density of HTB and its tailored design for bipartite graphs. In contrast, \textsf{Gorder} prioritizes the hit rate of the CPU cache and reorders all vertices of the entire graph, potentially leading to inefficiencies in our specific problem domain. }

\begin{table}
% \vspace{-1mm}
    \begin{center}   
        \caption{Time costs (sec.) of \textsf{GBC} on (un)reordered graphs ($p=q=8$).}  
        
        \vspace{-2mm}
        \label{vertexreordertable} 
        \begin{tabular}{|c||c|c|c|}
            \hline   \diagbox[width=2cm]{\textbf{Datasets}}{\textbf{Methods}} &  No Reorder & \textsf{Gorder} & \textsf{Border} \\ \hline
            \hline   \textit{YT} & 1.90 & 1.15 & \textbf{0.87} \\
            \hline   \textit{BC} & 14.00 & 3.29 & \textbf{2.99} \\
            \hline   \textit{GH} & 3.64 & 1.62 & \textbf{1.45} \\
            \hline   \textit{SO} & 0.15 & 0.04 & \textbf{0.03} \\
            \hline   \textit{YL} & 10.30 & 3.61 & \textbf{2.25} \\
            \hline   \textit{ID} & 5.64 & 2.19 & \textbf{1.61} \\
            \hline   \textit{LF} & INF & 17488.60 & \textbf{7753.15} \\
            \hline   \textit{FR} & 1251.01 & 946.36 & \textbf{669.98} \\
            \hline   \textit{S1} & 4.01 & 2.47 & \textbf{2.29} \\
            \hline   \textit{S2} & 4.65 & 2.59 & \textbf{2.36} \\
            \hline
        \end{tabular}
    \end{center}
    \vspace{-3mm}
\end{table}

\begin{table}[tbp]
    % \vspace{-2mm}
    \begin{center}   
        \caption{Time costs (sec.) of \textsf{GBC} varying load balancing strategies.} 
        \vspace{-2mm}
        \label{workstealtable} 
        \begin{tabular}{|c||c|c|c|c|c|c|}
           \hline 
            \diagbox[width=2.5cm]{\textbf{Methods}}{\textbf{Datasets}} & \textit{SO} & \textit{S2} & \textit{BC} & \textit{LF} & \textit{FR}\\ 
            \hline\hline
            No Balance & 0.14 & 5.38 & 8.02 & INF & 3941.55 \\
            \hline
            Pre-runtime Only & \textbf{0.03} & 2.52 & 3.19 & 9071.55 & 804.31 \\
            \hline
            Runtime Only & 0.05 & 3.86 & 7.83 & INF & 2967.25 \\
            \hline
            Joint & 0.07 & \textbf{2.36} & \textbf{2.99} & \textbf{7753.15} & \textbf{669.98} \\
            \hline
        \end{tabular}
    \end{center}
    % \vspace{-1mm}
\end{table}

% \vspace{-1mm}
{
\textbf{Effect of Workload Balancing.}
The joint workload balancing yields a substantial optimization effect with an average speedup of 2.2$\times$.
Table~\ref{workstealtable} further elaborates on the impact of different load balancing strategies.
First, both pre-runtime and runtime strategies accelerate the algorithm, demonstrating their effectiveness in balancing workloads.
Second, the pre-runtime strategy consistently outperforms the runtime strategy primarily due to its utilization of a fine-grained approach. This involves distributing vertices from the second layer to thread blocks, contrasting with the coarse-grained method of the runtime strategy, which dynamically reassigns unprocessed root vertices from occupied to idle blocks. Furthermore, the runtime strategy necessitates frequent global memory access for tasks such as block location and \textit{GCL} rewriting.
Lastly, the joint strategy demonstrates optimal performance in most scenarios, particularly under heavy workloads, underscoring the complementary nature of both strategies.}

{In conclusion, the DFS-BFS exploration, serving as the backbone, exhibits the highest average acceleration (3.7$\times$). 
Although the collective impact of HTB and \textsf{Border} demonstrate a comparable speedup of 2.2$\times$ to the workload balancing, their enhancement in performance gradually amplifies with larger datasets. For instance, on datasets \textit{GH}, \textit{SO}, and \textit{YL}, their average speedup reaches 2.5$\times$ (up to 3.1$\times$), whereas the average speedup of load balancing is 2.0$\times$ (up to 2.1$\times$). This progression is attributed to the increasing workload of intersection computation emerging as the performance bottleneck.}
\subsection{Evaluating Graph Partition}
\label{sec:graphpartitionevaluation}
% \vspace{-2mm}

Finally, we evaluate the effectiveness of \textsf{BCPar} through a comparative analysis with the widely-used graph partitioning algorithm, i.e., \textsf{METIS}~\cite{Karypis1998}, on the large dataset \textit{OR}.  
Given that \textsf{METIS} is originally designed for unipartite graphs, we construct an auxiliary graph to serve as input for \textsf{METIS}. This auxiliary graph contains the vertices within the selected layer, with pairwise connections established if and only if two vertices are mutual 2-hop neighbors.
{Figure~\ref{partitionresults} reports the average throughput on the partitioned graphs, i.e., the number of bicliques traversed per second.}
% For graphs generated by \textsf{METIS}, we measure the throughputs of counting intra-partition and inter-partition bicliques respectively. 
% But for graphs produced by \textsf{BCPar}, we just need to evaluate the performance of counting intra-partition bicliques. 
% This is because \textsf{BCPar} does not produce bicliques that span more than two partitions. 
 
% Figure~\ref{partitionresults}(a) reports the throughput on graphs partitioned by \textsf{METIS}. It's evident that the throughput of counting intra-partition bicliques is higher than counting inter-partition bicliques. This is because the GPU needs to repeatedly load partitioned graphs through PCIe to count bicliques spanning more than two partitions, which significantly decreases the whole throughput. 
{As shown in Figure~\ref{partitionresults}(a), the throughput achieved on the graph partitioned by \textsf{BCPar} consistently surpasses its \textsf{METIS} counterpart.
The superiority of \textsf{BCPar} arises from its ability to partition a graph into mutually autonomous subgraphs, thereby confining the search within the current subgraph, avoiding the need for introducing communication overheads.} Moreover, \textsf{BCPar} requires the vertices in the same partition to share as many common neighbors as possible, which facilitates load balancing. 
In contrast, the subgraphs partitioned by \textsf{METIS} introduce dependencies among them, necessitating frequent data transfers via low-bandwidth PCIe to enumerate bicliques across partitions, thereby diminishing the overall throughput. 
Figure~\ref{partitionresults}(b) illustrates the detailed throughputs concerning bicliques found in the same partition (referred to intra-partition) and those spanning across partitions (referred to inter-partition). {Undoubtedly, the throughputs for enumerating inter-partition bicliques are markedly inferior to those for enumerating intra-partition bicliques, constituting the bottleneck for METIS. Conversely, \textsf{BCPar} mitigates this bottleneck by incorporating the structural properties of bicliques.}

% We average the throughput of inter-partition and intra-partition bicliques on graphs partitioned by \textsf{METIS} as its total throughput. 
% And compare the throughput on graphs produced by \textsf{METIS} and \textsf{BCPar} respectively.  
% It is evident that the throughput on the graphs partitioned by \textsf{BCPar} is higher than those partitioned by \textsf{METIS}. This is because on the graphs partitioned by \textsf{METIS} the GPU needs to repeatedly load partitioned graphs through PCIe to count bicliques spanning more than two partitions, which significantly decreases the whole throughput. 
% In contrast, \textsf{BCPar} will not produce bicliques spanning two or more partitions, which effectively avoids the data load cost. 
% {We show the detailed throughputs of \textsf{METIS} in Figure~\ref{partitionresults}(b). The throughputs of counting inter-partition bicliques are significantly lower than counting intra-partition bicliques, which causes the low throughputs of \textsf{GBC} on graphs partitioned by \textsf{METIS}. }

\begin{figure}[!tbp]
    \centering
    % \hspace{-2mm}
    
    % \subfigure[Comparison of throughputs. ]{
    % \includegraphics[width=0.225\textwidth]{figures/partitioncompare.pdf}}
    % \vspace{-1mm}
    % \subfigure[On graphs generated by \textsf{METIS}.]{
    % \includegraphics[width=0.225\textwidth]{figures/metisthroughput.pdf}}\
    % \vspace{-3mm}
    % \hspace{-2.3mm}
    \hspace{-2mm}
    \includegraphics[width=0.49\textwidth]{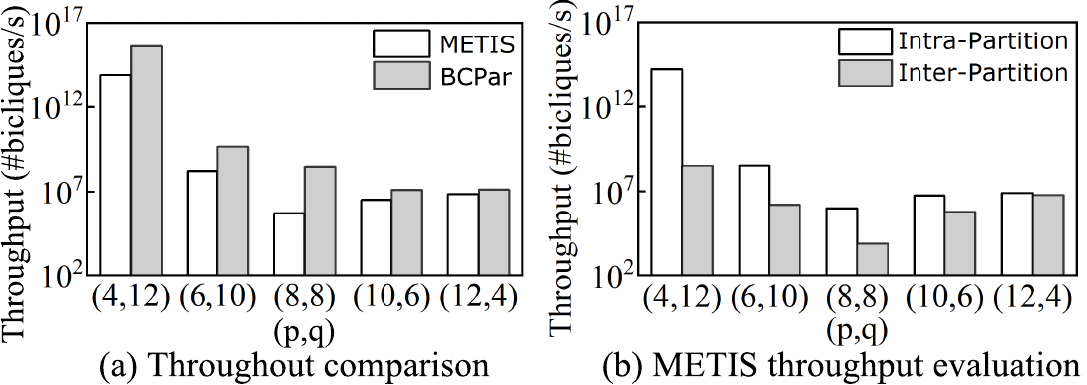}
    \vspace{-6mm}
    \caption{Throughput evaluation vs. different methods.}
    
    \label{partitionresults}
    % \vspace{1mm}
\end{figure}

\section{Related Work}
\label{sec:relatedwork}
% \vspace{-1.5mm}

{\textbf{CPU-based Motif Counting on Bipartite Graphs.} 
Wang et al.~\cite{wang2014rectangle} initiate the exploration of butterfly counting and introduce wedge-based counting algorithms. Sanei-Mehri et al.~\cite{sanei2018butterfly} develop a random algorithm to estimate butterflies, while Wang et al.~\cite{sanei2018butterfly} formulate a vertex-priority strategy for butterfly listing. Recent studies extend butterfly counting to streaming graphs~\cite{sanei2019fleet, sheshbolouki2022sgrapp}, uncertain graphs~\cite{sanei2019fleet, sheshbolouki2022sgrapp}, and temporal graphs~\cite{cai2023efficient}.
Over the years, several algorithms have been proposed for identifying maximal bicliques~\cite{ li2007maximal,chen2022efficient}.
% Identifying maximal bicliques, another thoroughly investigated problem on bipartite graphs, several algorithms have been proposed over the years~\cite{ li2007maximal,chen2022efficient}. 
Recently, maximal bicliques in signed bipartite graphs~\cite{sun2022maximal}, uncertain graphs~\cite{wang2023efficient}, and those with unique properties like fairness-aware maximal cliques~\cite{yin2023fairness} have been explored.
Yang et al.~\cite{yang2021p} pioneers the $(p, q)$-biclique enumeration algorithm by introducing a recursive backtracking algorithm.}

{\textbf{GPU-accelerated motif counting on bipartite and unipartite graphs.} 
Xu et al.~\cite{xu2022efficient} pioneer GPU acceleration for butterfly counting, introducing a lock-free strategy to reduce synchronization and an adaptive strategy for workload balance.
In unipartite graph motifs, triangle counting on GPUs~\cite{Polak16} has seen various optimized algorithms, including workload estimation~\cite{green2018logarithmic}, bitmap-based intersection~\cite{BissonF17}, and wedge-oriented methods~\cite{hu2019triangle}, alongside  preprocessing techniques~\cite{Hu0L21}. 
% Maximal clique enumeration on GPUs has been extensively studied~\cite{almasri2022parallelizing,wei2021accelerating}.
% Almasri et al.~\cite{almasri2022parallel} achieves GPU-accelerated $k$-clique counting with two parallel acceleration methods based on graph orientation and pivoting. 
In addition, there exists a wealth of literature on GPU acceleration for maximal clique enumeration~\cite{almasri2022parallelizing,wei2021accelerating}, $k$-clique counting~\cite{almasri2022parallel}, $k$-core decomposition~\cite{tripathy2018scalable, mehrafsa2020vectorising, ahmad2023accelerating} and $k$-truss decomposition~\cite{almasri2019update, diab2020ktrussexplorer, che2020accelerating}. 
To the best of our knowledge, accelerating $(p,q)$-biclique counting on GPUs has not been thoroughly investigated in the existing literature. }
% \vspace{-2mm}
\section{Conclusion}
\label{sec:conclusion}
% \vspace{-1.5mm}

We introduce \textsf{GBC}, a novel GPU-based approach for counting $(p,q)$-bicliques on large bipartite graphs.  
{\textsf{GBC} enhances parallelism by adaptively consolidating tasks across multiple vertices during DFS.}
To facilitate efficient set intersection, we propose HTB, a novel data structure that reduces redundant comparisons and memory transactions. We further devise \textsf{Border} to compress HTB by reordering vertices. 
For scalability, we develop \textsf{BCPar} for handling large bipartite graphs beyond the GPU memory. 
Experimental results consistently demonstrate that \textsf{GBC} significantly surpasses existing algorithms, affirming its effectiveness, efficiency, and scalability.

% \vspace{-1mm}
\section{Acknowledgements}
% \vspace{-1mm}

{This work was supported in part by the NSFC under Grants No. (62025206, U23A20296, 62102351)}, Zhejiang Province's ``Lingyan'' R\&D Project under Grant No. 2024C01259, and Yongjiang Talent Introduction Programme (2022A-237-G). Yunjun Gao is the corresponding author of the work.

% \balance

\newpage
\balance
\clearpage

\bibliographystyle{abbrv}
\bibliography{ref}
% \balance

\newpage
\nobalance
% \clearpage
\appendix
\subsection{Breakdown of Running Time}
\label{fig:runtime_breakdown}

We provide a comprehensive breakdown of the time spent on various components, including HTB transformation, \textsf{Border} execution, and biclique search. Table~\ref{tab:time_component} presents the experimental outcomes.
Notably, the time required for HTB transformations is minimal, typically ranging from tens to hundreds of milliseconds, {generally falling below one percent, even one ten-thousandth of the counting time, and proportional to the number of vertices.}
With the additional time cost for \textsf{Border} reordering, which typically ranges from 0.18s to 62.17s, we observe that the overall runtime decreases by up to 4.60$\times$, highlighting the significant benefit of reordering in biclique searching.
Furthermore, it's important to note that \textsf{Border} can be reused for different $(p,q)$ parameters. Hence the amortized runtime of \textsf{Border} becomes practically negligible.

% \vspace{-4mm}
\begin{table}[htbp]
    \begin{center}   
        \caption{Time costs (sec.) of each component in \textsf{GBC}.}  
        \setlength{\tabcolsep}{2.3pt}
        \vspace{-2mm}
        \label{tab:time_component} 
        \begin{tabular}{|c||c|c|c|}
            \hline   \diagbox[width=2.3cm]{\textbf{Datasets}}{\textbf{Components}} &  HTB transformation &  Reorder & Counting \\ \hline
            \hline   \textit{YT} & 0.001 &0.18  &0.87    \\
            \hline   \textit{BC} & 0.008 &1.47  &2.99     \\
            \hline   \textit{GH} & 0.001 &0.46  &1.45     \\
            \hline   \textit{SO} & 0.007 &0.88  &0.03     \\
            \hline   \textit{YL} & 0.013 &2.96  &2.25     \\
            \hline   \textit{ID} & 0.031 &8.87  &1.61     \\
            \hline   \textit{LF} & 0.172 &62.17  &7753.15\\
            \hline   \textit{FR} & 0.007 &43.74  &669.98  \\
            \hline   \textit{S1} & $6\times10^{-4}$ &0.45  &2.29     \\
            \hline   \textit{S2} & $6\times10^{-4}$ &0.45  &2.36     \\
            \hline
        \end{tabular}
    \end{center}
    \vspace{-3mm}
\end{table}

\subsection{Comparison Between DFS and DFS-BFS}
We conduct a comparison of the memory consumption between the DFS-BFS and DFS methods, as depicted in Figure~\ref{fig:dfs_hybrid}.
On average, DFS-BFS incurs 1.3$\times$ more memory overhead, remaining well below the GPU memory capacity (hundreds of megabytes compared to 24 gigabytes of GPU memory).
However, DFS-BFS demonstrates superior performance, being on average 2.2$\times$ faster than DFS, attributed to the effective utilization of parallel threads offered by GPU.

\vspace{-2mm}
\begin{figure}[htbp]
    \centering
    \includegraphics[width=0.3\textwidth]{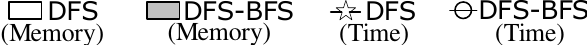}\\
    \hspace{-2mm}
    \includegraphics[width=0.27\textwidth]{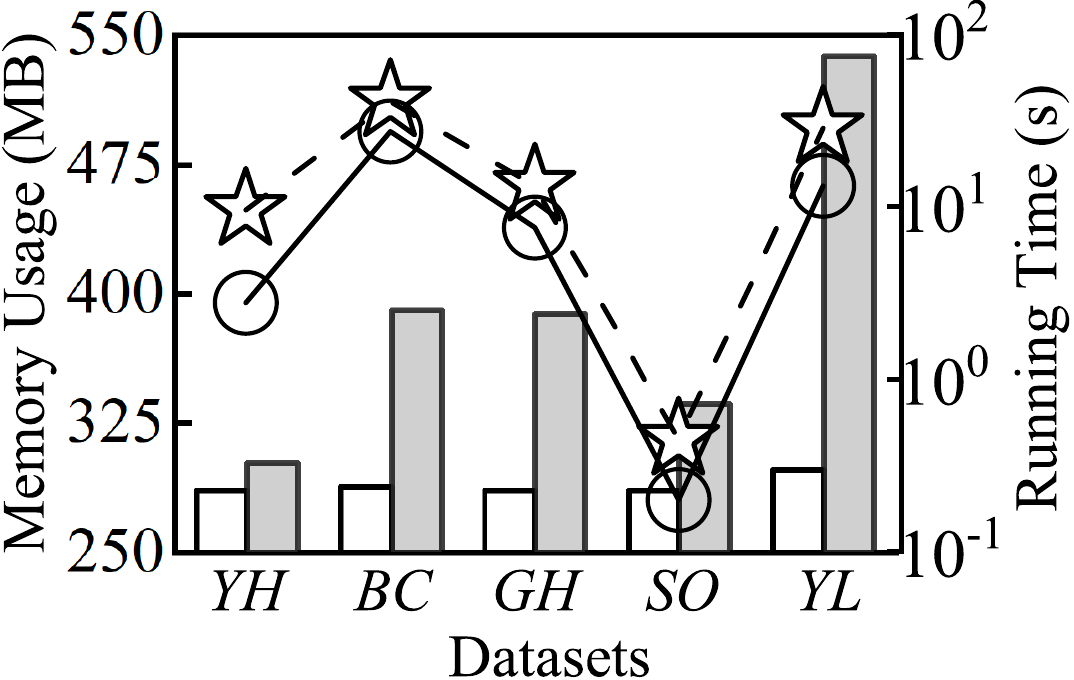}
    \vspace{-2mm}
    \caption{Performance between DFS and DFS-BFS}
    \label{fig:dfs_hybrid}
    % \vspace{-2mm}
\end{figure}
% \nobalance

\end{document}